\documentclass[11pt]{article}
\usepackage[letterpaper, left=1.25in, right=1.25in, top=1.25in, bottom=1.25in]{geometry}
\usepackage[rm,bf]{titlesec}
\usepackage{enumitem}
\usepackage{array}

\usepackage{amsmath}
\usepackage{amsthm}
\usepackage{multirow}
\usepackage{booktabs}

\usepackage{graphicx}
\usepackage{amssymb}
\usepackage{amsfonts}
\usepackage{amsmath,amsthm}
\usepackage[footnotesize,bf]{caption}
\usepackage{natbib}
\usepackage{setspace}
\usepackage[dvipsnames,usenames]{color}
\usepackage{pstricks,pstricks-add}
\usepackage{url}
\usepackage[usenames]{color}
\usepackage[breaklinks, colorlinks, citecolor=blue, linkcolor=blue, urlcolor=blue]{hyperref}
\usepackage{lscape}

\newtheorem{theorem}{Theorem}

\newtheorem{theorem*}[theorem]{Theorem*}

\newtheorem{definition}{Definition}
\newtheorem{definition*}[definition]{Definition*}

\newtheorem{example}{Example}
\newtheorem{example*}[example]{Example*}

\newtheorem{lemma}{Lemma}
\newtheorem{lemma*}[lemma]{Lemma*}

\newtheorem{proposition}{Proposition}
\newtheorem{proposition*}[proposition]{Proposition*}

\newcommand{\Mcal}{\mathcal{M}}
\newcommand{\Bcal}{\mathcal{B}}

\newcommand{\Ucal}{\mathcal{U}}

\newcommand{\N}{\mathbb{N}}
\newcommand{\R}{\mathbb{R}}

\newcommand{\BNE}{\textrm{BNE}}
\newcommand{\Nash}{\textrm{N}_{eq}}
\newcommand{\UNash}{\textrm{U}_{eq}}
\newcommand{\TNash}{\textrm{T}_{eq}}
\newcommand{\PNash}{\textrm{P}_{eq}}
\newcommand{\ENash}{\textrm{E}_{eq}}

\newcommand{\RNash}{\textrm{R}}

\newcommand{\exputil}{\textrm{E}}
\newpsobject{grilla}{psgrid}{subgriddiv=1,griddots=10,gridlabels=6pt}

\begin{document}
\bibliographystyle{elsart-harv}
\title{Empirical equilibrium\footnote{This paper benefited from comments of Dustin Beckett, Antonio Cabrales, Thomas Palfrey, Brian Rogers, and Anastasia Zervou as well as audiences at the 2018 Naples Workshop on Equilibrium Analysis, the 2019 World Economic Science Association Meetings, the 6th annual Texas Experimental Association Symposium, Boston College, Chapman U., Kellogg (MEDS), NCSU, Ohio State University, TETC18, UCSD, U. Maryland, U. Rochester, UTDallas and U. Virginia.  Sec.~\ref{Sec:Harsanyapproach} of this paper subsumes the results circulated in a note studying the empirical content of monotone randomly disturbed payoff models \citep{Velez-Brown-2019-Harsanyi}. An earlier version of this paper contained an application of empirical equilibrium to implementation theory. The implementation results are now part of a companion paper \citep{Velez-Brown-2018-EI}.   All errors are our own.}}
\date{\today}

\author{Rodrigo A. Velez\thanks{
\href{mailto:rvelezca@tamu.edu}{rvelezca@tamu.edu}; \href{https://sites.google.com/site/rodrigoavelezswebpage/home}{https://sites.google.com/site/rodrigoavelezswebpage/home}}\ \ and Alexander L. Brown\thanks{
 \href{mailto:alexbrown@tamu.edu}{alexbrown@tamu.edu}; \href{http://people.tamu.edu/\%7Ealexbrown}{http://people.tamu.edu/$\sim$alexbrown}} \\\small{\textit{Department of
Economics, Texas A\&M University, College Station, TX 77843}}}
\maketitle

\begin{abstract}
We study the foundations of \textit{empirical equilibrium}, a refinement of Nash equilibrium that is based on a non-parametric characterization of empirical distributions of behavior in games \citep{Velez-Brown-2019-SP}. The refinement can be alternatively defined as those Nash equilibria that do not refute the regular QRE theory of \citet{Goeree-Holt-Palfrey-2005-EE}. By contrast, some empirical equilibria may refute monotone additive randomly disturbed payoff models. As a by product, we show that empirical equilibrium does not coincide with refinements based on approximation by monotone additive randomly disturbed payoff models, and further our understanding of the empirical content of these models.
\medskip
\begin{singlespace}

\medskip

\medskip

\end{singlespace}
\end{abstract}

\section{Introduction}\label{Sec:intro}


We study the foundations of \textit{empirical equilibrium}, a refinement of Nash equilibrium in normal-form games introduced by \citet{Velez-Brown-2019-SP}. We provide two alternative characterizations of the set of empirical equilibria of a game, and determine its relationship with previous refinements in the literature. The first characterization eases its computation in applications. The second characterization provides a direct connection between one of the most popular models for the analysis of data from economics experiments and this refinement.

It is well known that the set of Nash equilibria of a game may contain elements that are implausible. That is, it is somehow evident that some Nash equilibria are not likely to be observed if the game actually takes place (see Example~\ref{Ex:Gamma_1}).  If one is interested in the positive content of this theory, this is problematic. In applications, one often performs extreme case scenario analyses based on the set of outcomes predicted for a certain game by a solution concept. If some of these outcomes are implausible, then a worst case scenario analysis can be unnecessarily pessimistic, and a best case scenario analysis can be unrealistically optimistic.

The response of game theory and economics to this problem has been to advance equilibrium refinements, i.e., selections from the Nash equilibrium solution. With only few exceptions---which have never been used in applications nor further studied (details below)---equilibrium refinements determine as implausible each equilibrium in which an agent plays a weakly dominated action with positive probability (from \citealp{Selten-1975-IJGT} and \citealp{Myerson-1978-IJGT}, to \citealp{K-Mertens-1986-Eca} and the latest iterations in \citealp{Milgrom-Mollner-2017-SSRN} and \citealp{Fudenberg-He-2018}; see \citealp{VanDamme-1991-Springer} for an earlier survey).  By contrast, experimental studies have shown that weakly dominated behavior is persistently observed in strategic situations.  In particular, a robust body of evidence has emerged from experiments in games that have dominant strategy equilibria induced by strategy-proof mechanisms \cite[see][for a multiple study analysis]{Velez-Brown-2019-SP}.

Motivated by the gap between theory and data, \citet{Velez-Brown-2019-SP} proposed empirical equilibrium, a refinement of Nash equilibrium that is based solely on observables. This refinement produces a delicate selection from the Nash equilibrium set that does not discard all weakly dominated behavior. It produces sharp predictions for dominant strategy games \citep{Velez-Brown-2019-SP} and partnership dissolution auctions \citep{Velez-Brown-2018-EI}. It allows one to come to terms with experimental evidence on these games \citep{Brown-Velez-2019-TBB}.

Empirical equilibrium is defined by means of the following thought experiment. Consider a researcher who samples behavior in normal-form games and constructs a theory that explains this behavior. The researcher determines the plausibility of Nash equilibria based on the empirical content of the theory by requiring that Nash equilibria be in its closure. That is, if a Nash equilibrium cannot be approximated to an arbitrary degree by the empirical content of the researcher's theory, it is identified as implausible or unlikely to be observed. Empirical equilibrium is based on the non-parametric theory that each agent chooses actions with higher probability only when they are better for her given what the other agents are doing. More precisely, an empirical equilibrium is a Nash equilibrium that is the limit of a sequence of behavior satisfying \emph{weak payoff monotonicity}, i.e., between two alternative actions for an agent, say $a$ and $b$, if the agent plays $a$ with higher frequency than $b$, it is because given what the other agents are doing, $a$ has higher expected utility than $b$.

One can give  empirical equilibrium a static or dynamic interpretation. First, its definition simply articulates the logical implication of the hypothesis that the researcher's theory is well specified, for this hypothesis is refuted by the observation of a Nash equilibrium that is not an empirical equilibrium. Alternatively, suppose that the researcher hypothesizes that behavior will converge to a Nash equilibrium through an unmodeled evolutionary process that produces a path of behavior that is consistent with her theory. Then, the researcher can also conclude that the only Nash equilibria that will be approximated by behavior are empirical equilibria.

Thus, the essential component in the definition of empirical equilibrium is the closure of the empirical content of weak payoff monotonicity, the theory on which its definition is based. In our main results we study the relationship between this empirical content and that of alternative non-parametric and parametric theories.

Our first result is that empirical equilibrium can be alternatively defined by proximity to full-support \textit{payoff monotone behavior}, i.e., each interior distribution in which probabilities are ordinally equivalent to the expected utility vectors it induces (Theorem~\ref{Thm:EE=APP-interior}). This result is interesting in two dimensions. On the one hand, since payoff monotone behavior is also weakly payoff monotone, this characterization simplifies the computation of the empirical equilibrium set in applications. For instance, this result is crucial for the computation of the empirical equilibrium set in partnership dissolution auctions \citep{Velez-Brown-2018-EI}. On the other hand, it shows a robustness property of empirical equilibria. Even though weak payoff monotonicity is intuitive, it may accept behavior in which actions that have different expected utility are played with the same probability. Thus, if one is suspicious of the plausibility of this type of behavior, our theorem guarantees the approximation of an empirical equilibrium can always be done by means of payoff monotone behavior.

Empirical equilibrium is related with \textit{firm} equilibria, which is defined by means of approximation by exchangeable randomly disturbed payoff models \citep{VanDamme-1991-Springer}; \textit{vanishing control cost approachable} equilibria, which is defined by means of approximation in control cost games \citep{VanDamme-1991-Springer};  and \textit{logistic QRE approachable} equilibria \citep{Mackelvey-Palfrey-1996-JER}. Though not considered before in the literature, we could also define a refinement by approximation by means of behavior in \textit{structural QRE} \citep{mckelvey:95geb}, and \textit{regular QRE} \citep{Mackelvey-Palfrey-1996-JER,Goeree-Holt-Palfrey-2005-EE}, two popular theories for the analysis of data in experiments. Each of these theories generate interior payoff monotone behavior. Theorem~\ref{Thm:EE=APP-interior} alerts us about the possibility that empirical equilibrium coincides with some of these alternative refinements. The remaining results in the paper determine if this is so.

The empirical content of control cost games is contained in that of the regular QRE theory \citep{VanDamme-1991-Springer,Goeree-Holt-Palfrey-2016-Book}. It turns out that empirical equilibrium can be defined based on the empirical content of a finitely dimensional subfamily of control cost games (Lemma~\ref{Lm:cost=mon}). Thus, it can also be defined by means of approximation by regular QRE behavior (Lemma~\ref{Lm:Mon=rQRE}).

Regular QRE also allows us to articulate the idea of increasingly sophisticated behavior. That is, one can identify sequences of regular QRE that \textit{are utility maximizing in the limit}, i.e., that generate sequences of behavior that can accumulate only on mutual best responses. Remarkably, each empirical equilibrium in a game is the limit of behavior generated by a sequence of regular QRE that is utility maximizing in the limit (Theorems~\ref{Thm:EE=R} and~\ref{Th:E=C}).\footnote{It is not clear that vanishing control cost approachable equilibrium coincides with empirical equilibrium. This equilibrium refinement requires control cost functions vanish in a particular parametric form that is more restrictive than the general form used to prove Theorems~\ref{Thm:EE=R} and~\ref{Th:E=C}.}
 Regular QRE has been successful generating key comparative statics in data from diverse strategic environments \citep{Goeree-Holt-Palfrey-2016-Book}.\footnote{Behavior in the empirical implementation of the popular cognitive hierarchy models of \citet{Camerer-et-al-QJE-2004} can be arbitrarily approximated by regular QRE with heterogeneous non-common prior beliefs \citep{ROGERS-et-al-2009-JET}.} Even though it is not universal, estimates of parameters of the most popular form of regular QRE, its logistic form, tend to move toward its infinitely sophisticated extreme in experiments that are repeated multiple rounds \citep{mckelvey:95geb,Goeree-Holt-Palfrey-2016-Book}. Thus, Theorems~\ref{Thm:EE=R} and~\ref{Th:E=C} provide us with a direct connection between the practice in experimental economics, the hypotheses that these researchers usually adopt and test in their analysis, and empirical equilibrium.

In contrast to the universal possibility of approximation of empirical equilibria by regular QRE, there are empirical equilibria that cannot be approximated by the empirical content of logistic QRE, structural QRE, and exchangeable randomly disturbed payoff models (Sec.~\ref{Sec:Harsanyapproach}).\footnote{Thus, one should not be surprised if behavior in an experiment moves toward best responses while logistic QRE parameter estimates do not move towards their sophisticated extreme.} Indeed, if at least an agent has at least three actions available, one can always construct a payoff matrix for which there is an empirical equilibrium that cannot be approximated by the empirical content of a general form of monotone randomly disturbed payoff models introduced by \citet{GOVINDAN-Reny-Robson-2003-GEB} and whose empirical content includes that of these three models (Theorem~\ref{Thm:Paradox2}). This is surprising. Randomly disturbed payoff models are not refutable if their parameters are left unrestricted in the domain in which they are defined \citep{Haile-et-al-2008}. Moreover, even if one also disciplines them with convergence to best response operators, they induce no proper refinement of the Nash equilibrium set (Proposition~\ref{Prop:unrefutability}). However, if they are disciplined by payoff monotonicity, their empirical content is constrained beyond monotonicity itself.

Despite their popularity in experimental research and the central role that randomly disturbed payoff models have in the foundations of Nash equilibrium, very little was known about their empirical content when they are disciplined by monotonicity. In particular, \citet{Goeree-Holt-Palfrey-2005-EE} posed the question whether the empirical content of structural QRE (for a fixed game in which at least an agent has at least three actions available) coincides with that of the regular QRE models. A byproduct of our results is a sharp answer to this question, i.e., that the empirical contents of these theories differ in fundamental ways. On the one hand, regular QRE, in particular its parametric form in control cost games, generate the same empirical content as payoff monotonicity. Thus, the empirical content of regular QRE, for a fixed game, is not restricted by any of the additional assumptions on the unobservables that define the theory (see Sec.~\ref{Sec:rQRE}). On the other hand, monotone randomly disturbed payoff models can be misspecified for the study of behavior that not only is consistent with payoff monotonicity, but is also disciplined by increasing sophistication.

The remainder of the paper is organized as follows. Sec.~\ref{Sec:Model} introduces the model and definitions. Sec.~\ref{Sec:Examples} presents a series of examples that allow the reader to familiarize with empirical equilibrium and show that this refinement is independent from the tremble based refinements previously defined in the literature. Sec.~\ref{Sec:Foundations} presents our results. Sec.~\ref{Sec:Discussion} discusses and concludes.

\section{Model}\label{Sec:Model}

We study the plausibility of Nash equilibria in a finite normal-form game $\Gamma(u):=(N,A,u)$ where $N:=\{1,...,n\}$ is a set of agents; $(A_i)_{i\in N}$ are the corresponding action spaces and $A:= A_1\times\dots\times A_n$ the set of action profiles; and  $u:=(u_i)_{i\in N}$ is the profile of expected utility indices, i.e., functions $u_i:A\rightarrow \R$. Let $\Ucal$ be the set of all utility profiles. Our interpretation of the game is standard. Agents simultaneously choose an action. Given that action profile $a:=(a_i)_{i\in N}$ is chosen, agent~$i$'s payoff is $u_i(a)$. Our analysis will not involve comparisons of behavior across games with different agent sets or action spaces. Thus, $N$ and $A$ are fixed throughout.

A strategy for agent $i$  is a probability distribution on $A_i$, denoted generically by $\sigma_i\in\Delta(A_i)$. A pure strategy places probability one on a given action. We identify pure strategies with the actions themselves. A strategy is interior if it places positive probability on each possible action. A profile of strategies is denoted by $\sigma:=(\sigma_i)_{i\in N}\in\Sigma(A):=\Delta(A_1)\times\dots\times\Delta(A_n)$. Given $S\subseteq N$, we denote a subprofile of strategies for these agents by $\sigma_S$. When $S=N\setminus\{i\}$, we simply write $\sigma_{-i}\in \Sigma(A)_{-i}:=\times_{j\in N\setminus\{i\}}\Delta(A_j)$. Consistently, we concatenate partial strategy profiles as in $(\sigma_{-i},\mu_i)$. We consistently use this convention when operating with vectors, as with action profiles.

Agent $i$'s expected utility given strategy profile $\sigma$ is
\[\exputil_\sigma u_i=\sum_{a\in A}u_i(a)\sigma(a),\]
where $\sigma(a)=\sigma_1(a_1)\dots\sigma_n(a_n)$. Following our convention of identifying pure strategies with actions, we write $\exputil_{(\sigma_{-i},a_i)}u_i$ for the utility that agent $i$ gets from playing actions $a_i$ when the other agents play $\sigma_{-i}$. We say that an action $a_i\in A_i$ is weakly dominated by action $\hat a_i \in A_i$ if for each $a_{-i}\in A_{-i}$, $u_i(a_{-i},\hat a_i)\geq u_i(a)$ with strict inequality for at least an element of $A_{-i}$. We say that $a_i\in A_i$  is a weakly dominated action if there is another action that weakly dominates it.

The following are the basic prediction for game $\Gamma(u)$ and three of its most prominent refinements.
\begin{enumerate}
\item \citep{Nash-1951} A \textit{Nash equilibrium of $\Gamma(u)$} is a profile of strategies $\sigma$, such that for each $i\in N$ and each $\mu_i\in\Delta(A_i)$, $\exputil_\sigma u_i\geq \exputil_{(\sigma_{-i},\mu_i)}u_i$ . We denote this set by $\Nash(\Gamma(u))$.
\item An \textit{undominated Nash equilibrium of $\Gamma(u)$} is a Nash equilibrium of $\Gamma$ in which no agent plays with positive probability a weakly dominated action. We denote this set by $\UNash(\Gamma(u))$.

\item \citep{Selten-1975-IJGT}   A \textit{perfect equilibrium} of $\Gamma$ is a profile of strategies $\sigma$ that is the limit of a sequence of interior strategy profiles $\{\sigma^\lambda\}_{\lambda\in\N}$ such that for each $\lambda\in \N$ and each $i\in N$, $\sigma^\lambda_i$ places probability greater than $1/\lambda$ on a given action only if it is a best response to $\sigma^\lambda_{-i}$. We denote this set by $\TNash(\Gamma(u))$.\footnote{Our definition of perfect equilibrium corresponds to \citet{Myerson-1978-IJGT}'s characterization of \citet{Selten-1975-IJGT}'s perfect equilibrium.}

\item \citep{Myerson-1978-IJGT} A \textit{proper equilibrium} of $\Gamma(u)$  is a profile of strategies $\sigma$ that is the limit of a sequence of interior strategy profiles $\{\sigma^\lambda\}_{\lambda\in\N}$ such that for each $\lambda\in \N$, each $i\in N$, and each pair of actions $\{a_i,\hat a_i\}\subseteq A_i$, if $\exputil_{(\sigma^\lambda_{-i},a_i)}u_i>\exputil_{(\sigma^\lambda_{-i},\hat a_i)}u_i$, then $\sigma^\lambda_i(\hat a_i)\leq (1/\lambda)\sigma^\lambda_i(a_i)$. We denote this set by $\PNash(\Gamma(u))$.
\end{enumerate}

Our main objective is to study empirical equilibrium, a refinement of Nash equilibrium that is based on an empirical characterization of behavior. That is, we envision that the researcher samples empirical distributions of behavior in a finite set of normal-form games.  Based on the analysis of the data, the researcher constructs a refutable theory that explains this behavior. Then, uses this theory to determine the plausibility of Nash equilibria in all normal-form games. If a Nash equilibrium is not in the closure of the empirical content of the researchers theory, the researcher would be able to reject the specification of the theory were he or she to observe this equilibrium. Thus, under the hypothesis that the researcher's theory is well specified, each Nash equilibrium that does not belong to the closure of the empirical content of the researcher's theory is implausible. Empirical equilibrium is the refinement so defined when the researcher endorses the non-parametric theory that each agent chooses actions with higher probability only when they are better for her given what the other agents are doing.

\begin{definition}[\citealp{Velez-Brown-2019-SP}]\rm $\sigma\in\Sigma(A)$ is \textit{weakly payoff monotone for $u$} if for each $i\in N$ and each pair of actions $\{a_i,\hat a_i\}\subseteq A_i$ such that $\sigma_i(a_i)>\sigma_i(\hat a_i)$, we have that $\exputil_{(\sigma_{-i},a_i)}u_i>\exputil_{(\sigma_{-i},\hat a_i)}u_i$.
\end{definition}

Intuitively, a profile of strategies is weakly payoff monotone for a game if differences in behavior reveal differences in expected payoffs.

\begin{definition}[\citealp{Velez-Brown-2019-SP}]\rm An \textit{empirical equilibrium of $\Gamma(u)$} is a Nash equilibrium of $\Gamma(u)$ that is the limit of a  sequence of weakly payoff monotone strategies for $u$. We denote this set by $\ENash(\Gamma(u))$.
\end{definition}

It is well-known that proper equilibria are Nash equilibria. One can easily see that proper equilibria are empirical equilibria. For each finite game the set of Proper equilibria is non-empty \citep{Myerson-1978-IJGT}. Thus, for each finite game the set of empirical equilibria is also non-empty \citep{Velez-Brown-2019-SP}.

The following property of a strategy profile, which implies weak payoff monotonicity, will allow us to provide an alternative useful characterization of empirical equilibrium.

\begin{definition}\rm $\sigma\in\Sigma(A)$ is \textit{payoff monotone for $u$} if for each $i\in N$ and each pair of actions $\{a_i,\hat a_i\}\subseteq A_i$, $\sigma_i(a_i)\geq \sigma_i(\hat a_i)$ if and only if $\exputil_{(\sigma_{-i},a_i)}u_i\geq \exputil_{(\sigma_{-i},\hat a_i)}u_i$.
\end{definition}

\section{Empirical equilibrium and tremble based refinements}\label{Sec:Examples}

In this section we study the relationship between empirical equilibrium and undominated and perfect equilibria. We do so by analyzing a series of examples showing that these equilibrium concepts are independent. The main purpose of this discussion is to provide the reader with clear intuition about empirical equilibrium by contrasting it with these more familiar refinements.\footnote{In Sec.~\ref{Sec:Foundations} we discuss two refinements introduced by \citet{VanDamme-1991-Springer} and \cite{mckelvey:95geb} that are subrefinements of empirical equilibrium. These studies provide examples showing that these refinements may not be contained in the set of undominated equilibria. Thus, one can conclude from the examples in \citet{VanDamme-1991-Springer} and \cite{mckelvey:95geb} that empirical equilibria may involve weakly dominated actions are played with positive probability.}

\begin{table}[t]
  \centering
  \begin{tabular}{ccp{2.4cm}p{2.5cm}}
  &    &\multicolumn{2}{c}{Player 2}
  \\
  &&  \multicolumn{1}{c}{$b_1$}&\multicolumn{1}{c}{$b_2$}
\\\cline{3-4}
&$a_1$&\multicolumn{1}{|c}{$1,1$}&\multicolumn{1}{|c|}{$0,0$}
\\\cline{3-4}
Player 1& $a_2$&\multicolumn{1}{|c}{$0,0$}&\multicolumn{1}{|c|}{$0,0$}
\\\cline{3-4}
\multicolumn{4}{c}{}
\\
\multicolumn{4}{c}{(a) $u^1$}  \end{tabular}
\begin{tabular}{c}
\\
\\
\\
\\
\end{tabular}
\begin{tabular}{ccp{2.4cm}p{2.5cm}}
  &    &\multicolumn{2}{c}{Player 2}
  \\
  &&  \multicolumn{1}{c}{$b_1$}&\multicolumn{1}{c}{$b_2$}
\\\cline{3-4}
&$a_1$&\multicolumn{1}{|c}{$2,2$}&\multicolumn{1}{|c|}{$2,1$}
\\\cline{3-4}
Player 1& $a_2$&\multicolumn{1}{|c}{$2,3$}&\multicolumn{1}{|c|}{$0,0$}
\\\cline{3-4}
\multicolumn{4}{c}{}
\\
\multicolumn{4}{c}{(b) $v^1$ }
  \end{tabular}
    \caption{In two games shown $N=\{1,2\}$, $A_1=\{a_1,a_2\}$, $A_2=\{b_1,b_2\}$, and payoffs are shown in the corresponding table; (a)  a game in which the set of empirical equilibria is a proper subset of the set of Nash equilibria; (b) a game in which there are empirical equilibria in which player 1 chooses a weakly dominated strategy with positive probability.}\label{Tab:Gamma1-and-Psi}
\end{table}

\begin{example}\label{Ex:Gamma_1}\rm
Consider game $\Gamma(u^1)$ in Table~\ref{Tab:Gamma1-and-Psi} (a). This game was proposed by \citet{Myerson-1978-IJGT} to illustrate that some Nash equilibria are intuitively implausible. There are two Nash equilibria in $\Gamma(u^1)$, $(a_1,b_1)$ and $(a_2,b_2)$. Only $(a_1,b_1)$ is an empirical equilibrium in this game. Indeed, for each distribution of actions of player 2, player 1's utility from playing $a_1$ is greater than or equal to the utility from playing $a_2$; thus, in a profile of weakly payoff monotone distributions of play, agent~$1$ will always play~$a_1$ with probability at least~$1/2$ (Fig.~\ref{Table-valuations-2} (a)); thus, $(a_2,b_2)$ cannot be approximated by weakly payoff monotone behavior. If this game is played and agents behavior is weakly payoff monotone and approximates a Nash equilibrium, it is necessarily $(a_1,b_1)$. $\qed$
\end{example}

Each refinement that rules out weakly dominated behavior coincides with empirical equilibrium in game $\Gamma(u^1)$. Undominated equilibria and empirical equilibria are independent, however.

\begin{example}\rm Consider game $\Gamma(v^1)$ in Table~\ref{Tab:Gamma1-and-Psi} (b). Player~$2$ has a strictly dominant strategy in this game. Thus, in each Nash equilibrium~$\sigma$ of~$\Gamma(v^1)$, $\sigma_2(b_1)=1$. Agent~$1$ is indifferent between both actions if agent~$2$ plays~$b_1$.  Thus, the set of Nash equilibria of this game is the distributions in which agent~$1$ randomizes between both actions and agent~$2$ plays~$b_1$. Now, let $\sigma$ be a weakly payoff monotone distribution for $\Gamma(v^1)$. Since $b_1$ strictly dominates $b_2$, $\sigma_2(b_1)\geq \sigma_2(b_2)$. If $\sigma_2(b_2)>0$, $E_{(\sigma_2,a_1)}v^1_1>E_{(\sigma_2,a_2)}v^1_1$. Thus, it must be the case that $\sigma_1(a_1)\geq\sigma_1(a_2)$. If $\sigma_2(b_2)=0$, $E_{(\sigma_2,a_1)}v^1_1=E_{(\sigma_2,a_1)}v^1_1$. Thus, $\sigma_1(a_1)=\sigma_1(a_2)$. Thus, the set of weakly payoff distributions for $\Gamma(v^1)$ are those at which $\sigma_1(a_1)\geq 1/2$ and $\sigma_2(b_1)\geq 1/2$, except those at which $\sigma_2(b_1)=1$ and $\sigma_1(a_1)<1/2$ (Fig.~\ref{Table-valuations-2} (b)). The set of empirical equilibria of $\Gamma(v^1)$ are the Nash equilibria in which agent $1$ plays $a_1$ with probability at least $1/2$. Since~$a_2$ is weakly dominated by~$a_1$ for player~$1$, almost all of these empirical equilibria involve one player playing a weakly dominated action with positive probability. $\qed$
\end{example}

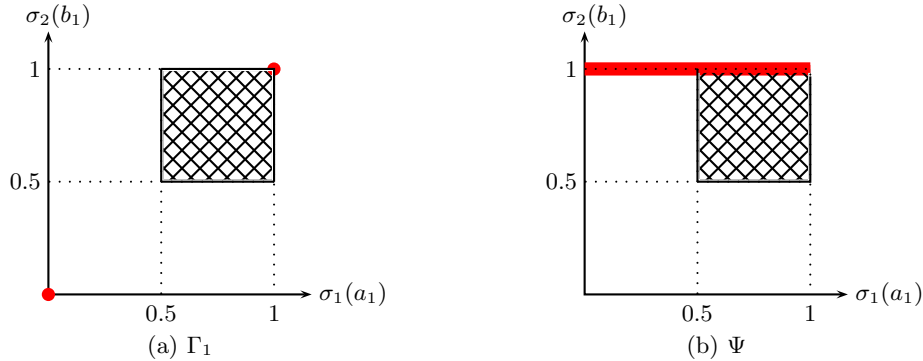
\begin{figure}[t]
\centering
\begin{pspicture}(6,2)(11,7)
\psline{<->}(6.5,6.5)(6.5,3)(10,3)
\psdots[dotsize=5pt,linecolor=red](6.5,3)(9.5,6)
\psframe[fillstyle=crosshatch](8,4.5)(9.5,6)
\psframe[linecolor=white](8,4.5)(9.5,6)
\psline[linestyle=dotted](6.5,6)(9.5,6)(9.5,3)
\psline[linestyle=dotted](8,3)(8,6)
\psline[linestyle=dotted](6.5,4.5)(9.5,4.5)
\psline(8,4.5)(8,6)
\psline(8,4.5)(9.5,4.5)
\psline(8,6)(9.5,6)(9.5,4.5)
\rput[l](10.1,3){$\mbox{\footnotesize $\sigma_1(a_1)$}$}
\rput[l](6.2,6.7){$\mbox{\footnotesize $\sigma_2(b_1)$}$}
\rput[c](8,2.75){$\mbox{\footnotesize $0.5$}$}
\rput[r](6.4,4.5){$\mbox{\footnotesize $0.5$}$}
\rput[c](9.5,2.75){$\mbox{\footnotesize $1$}$}
\rput[r](6.4,6){$\mbox{\footnotesize $1$}$}
\rput[c](8.25,2.3){\footnotesize (a) $\ENash(\Gamma(u^1))$}
\end{pspicture}
\begin{pspicture}(4,2)(11,7)
\psline{<->}(6.5,6.5)(6.5,3)(10,3)
\psline[linecolor=red,linewidth=5pt](6.5,6)(9.5,6)
\psframe[fillstyle=crosshatch](8,4.5)(9.5,6)
\psframe[linecolor=white](8,4.5)(9.5,6)
\psline[linecolor=red,linewidth=3pt](9.5,6)(8,6)
\psline[linestyle=dotted](6.5,6)(9.5,6)(9.5,3)
\psline[linestyle=dotted](8,3)(8,6)
\psline[linestyle=dotted](6.5,4.5)(9.5,4.5)
\psline(8,4.5)(8,6)
\psline(8,4.5)(9.5,4.5)(9.5,5.9)
\rput[l](10.1,3){$\mbox{\footnotesize $\sigma_1(a_1)$}$}
\rput[l](6.2,6.7){$\mbox{\footnotesize $\sigma_2(b_1)$}$}
\rput[c](8,2.75){$\mbox{\footnotesize $0.5$}$}
\rput[r](6.4,4.5){$\mbox{\footnotesize $0.5$}$}
\rput[c](9.5,2.75){$\mbox{\footnotesize $1$}$}
\rput[r](6.4,6){$\mbox{\footnotesize $1$}$}
\rput[c](8.25,2.3){\footnotesize (b) $\ENash(\Gamma(v^1))$}
\end{pspicture}
\caption{(a) Weakly payoff monotone distributions (shaded area) and Nash equilibria  of $\Gamma(u^1)$; $\sigma_1(a_1)$ is the probability with which agent $1$ plays $a_1$. Equilibrium $(a_1,b_1)$ can be approximated by weakly payoff monotone behavior. Thus, it is an empirical equilibrium of $\Gamma(u^1)$. Equilibrium $(a_2,b_2)$ cannot be approximated by weakly payoff monotone behavior. Thus, it is not an empirical equilibrium of $\Gamma(u^1)$. (b) Weakly payoff monotone distributions and Nash equilibria of $\Gamma(v^1)$; each Nash equilibrium in which agent $1$ plays $a_1$ with probability at least $1/2$ is an empirical equilibrium.}\label{Table-valuations-2}
\end{figure}

Empirical equilibrium does a subtle selection from the Nash equilibrium set. It determines the plausibility of a strategy based on its relative merits with respect to the alternative actions that the agent may choose. The following example drives this point home. It illustrates it for a parametric family of games. This family is a generalization of a game proposed by \citet{Myerson-1978-IJGT} to show that it is possible to introduce weakly dominated actions in $\Gamma(u^1)$, and considerably change its set of trembling hand perfect equilibria.

\begin{table}[t]
  \centering
  \begin{tabular}{ccp{2.5cm}p{2.5cm}p{2.5cm}}
  &    &\multicolumn{3}{c}{Player 2}
  \\
  &&  \multicolumn{1}{c}{$b_1$}&\multicolumn{1}{c}{$b_2$}&\multicolumn{1}{c}{$b_3$}
\\\cline{3-5}
&$a_1$&\multicolumn{1}{|c}{$1,1$}&\multicolumn{1}{|c|}{$0,0$}&\multicolumn{1}{c|}{$-7-c_1,-7-c_2$}
\\\cline{3-5}
Player 1& $a_2$&\multicolumn{1}{|c}{$0,0$}&\multicolumn{1}{|c|}{$0,0$}&\multicolumn{1}{c|}{$-7,-7$}
\\\cline{3-5}
& $a_3$&\multicolumn{1}{|c}{$-7-c_1,-7-c_2$}&\multicolumn{1}{|c|}{$-7,-7$}&\multicolumn{1}{c|}{$-7,-7$}
\\\cline{3-5}
  \end{tabular}
    \caption{$N=\{1,2\}$, $A_1=\{a_1,a_2,a_3\}$, $A_2=\{b_1,b_2,b_3\}$, and payoffs $u^c$ given in the table, where $c:=(c_1,c_2)$, $c_1>0$, and $c_2>0$.}\label{Tab:Gamma_c}
\end{table}

\begin{example}\label{Ex:Gamma_c}\rm Consider game $\Gamma(u^c)$ for some $c:=(c_1,c_2)$, $c_1>0$, and $c_2>0$ (Table~\ref{Tab:Gamma_c}). Standard arguments show that for each $c>0$,
  \[\begin{array}{l}\Nash(\Gamma(u^c))=\{(a_1,b_1), (a_2,b_2), (a_3,b_3)\},
  \\\TNash(\Gamma(u^c))=\UNash(\Gamma(u^c))=\{(a_1,b_1), (a_2,b_2)\},\\\PNash(\Gamma(u^c))=\{(a_1,b_1)\}.\end{array}\]
In contrast to these refinements, the empirical equilibrium set of $\Gamma(u^c)$ depends on $c$. First, note that for no $c>0$, $(a_3,b_3)$ is an empirical equilibrium of $\Gamma(u^c)$. This is so because~$a_2$ weakly dominates~$a_3$ for player~$1$. Thus, in any weakly payoff monotone distribution for $\Gamma(u^c)$, player $1$ plays $a_2$ with a probability that is at least the probability with which she plays action $a_3$. Thus, no sequence of weakly payoff monotone distributions for $\Gamma(u^c)$ converges to $(a_3,b_3)$. On the other hand for each $c>0$, $(a_1,b_1)\in\PNash(\Gamma(u^c))\subseteq \ENash(\Gamma(u^c))$.

Let us now examine the plausibility of $(a_2,b_2)$ in $\Gamma(u^c)$. Think of the payoffs in the game as dollar amounts. Consider first a small $c$, say $c_1\approx c_2\approx0.01$. Let $\sigma$ be an empirical distribution of play that approximates $(a_2,b_2)$. In such a situation, $E_{(\sigma_2,a_1)}u^c_1\approx 0>E_{(\sigma_2,a_3)}u^c_1\approx-7$ and $E_{(\sigma_2,b_1)}u^c_2\approx 0>E_{(\sigma_2,b_3)}u^c_2\approx-7$. Thus, if expected utility guides the choices of the players, one can expect that player $1$ will play $a_1$ at least as often as $a_3$, and player $2$ will play $b_1$ at least as often as $b_3$. If this is so, action $a_1$ will have a greater utility than action $a_2$ for player~$1$, and action $b_1$ will have a greater utility than action $b_2$ for player~$2$. Thus, if expected utility guides the choices of the agents, $\sigma$ will not be close to $(a_2,b_2)$. Thus, a plausible empirical distribution, i.e., one that is informed by expected utility for this game, will never be close to $(a_2,b_2)$.

Now, consider a large $c$, say $c_1\approx c_2\approx 100,\!000$. Again, if $\sigma$ is an empirical distribution of play that approximates $(a_2,b_2)$ and is guided by expected utility, player~$1$ will be playing $a_1$ at least as often as $a_3$, and player~$2$ will be playing~$b_1$ at least as often as $b_3$. In contrast with our previous case, it does not follow that necessarily action $a_1$ will have a greater utility than action $a_2$ for player~$1$, and action $b_1$ will have a greater utility than action $b_2$ for player~$2$. This will only happen if player $1$ is playing $a_1$ at least one hundred thousand times as often as $a_3$, and player $2$ is playing $b_1$ at least one hundred thousand times as often as $b_3$. Thus, it is possible that $\sigma$ is informed by expected utility, i.e., $\sigma_1(a_1)>\sigma_1(a_3)$ and $\sigma_2(b_1)>\sigma_2(b_3)$, and at the same time $E_{(\sigma_2,a_2)}u^c_1>E_{(\sigma_2,a_1)}u^c_1$, $E_{(\sigma_1,b_2)}u^c_2>E_{(\sigma_1,b_1)}u^c_2$, $\sigma_1(a_2)\approx1$, and $\sigma_2(b_2)\approx1$. Essentially, since the possible loss for player~$1$ from playing $a_1$ is about 100,000.00, player $1$ can be scared away from playing $a_1$ if player $2$ is playing $b_3$ more than once each 100,000 times she plays $b_1$. This is still compatible with $b_3$ being the worst alternative given what the other agent is doing.

These arguments can be easily formalized to show that
\[\ENash(\Gamma_c)=\left\{\begin{array}{ll}\{(a_1,b_1)\}&\textrm{if }\min\{c_1,c_2\}\leq 1,\\\{(a_1,b_1),(a_2,b_2)\}&\textrm{Otherwise. }\end{array}\right.\]
One cannot expect that if one brings these games to a laboratory setting or has the opportunity to collect field data on them, the threshold $\min\{c_1,c_2\}=1$ will be a good predictor of a structural change in the behavior of the agents. However, it is reasonable that behavior in this game will depend on the size of $c$, as empirical equilibrium predicts, i.e., equilibrium $(a_2,b_2)$ will be relevant only for high values of $c$. Undominated equilibria, perfect equilibria, and proper equilibria all miss this point. Undominated equilibrium and perfect equilibrium miss that when $c$ is too low, actions $a_2$ and $b_2$ are de facto ``weakly dominated'' when they are played with almost certainty. That is, if they were going to be played with probability close to one, actions $a_1$ and $b_1$, would be preferred for the respective players. Thus, we can rule this equilibrium out by means of the following observation. It is not reasonable that we will observe a distribution of play in which an agent is not playing her unique maximizer of utility with high probability, say more than random play.

Finally, proper equilibrium dismisses $(a_2,b_2)$ independently of $c$. Think of our example with high $c$. For $(a_2,b_2)$ to be a proper equilibrium of $\Gamma(u^c)$, for large $\lambda$ there must be a distribution of play $\sigma^\lambda$ satisfying two conditions: (i) $\sigma^\lambda$ is close to $(a_2,b_2)$, and thus $E_{(\sigma^\lambda_2,a_1)}u^c_1\approx 0>E_{(\sigma^\lambda_2,a_3)}u^c_1\approx-7$ and $E_{(\sigma^\lambda_2,b_1)}u^c_2\approx 0>E_{(\sigma^\lambda_2,b_3)}u^c_2\approx-7$; and (ii) $\sigma^\lambda_1(a_1)>\lambda\sigma^\lambda_1(a_3)$ and $\sigma^\lambda_2(b_1)>\lambda\sigma^\lambda_2(b_3)$. For distributions where $\lambda\geq 100,\!000$, $a_2$ and $b_2$ are not maximizing choices for players 1 and 2, respectively, meaning $(a_2,b_2)$ cannot be a proper equilibrium. Thus, the reason why proper equilibrium dismisses $(a_2,b_2)$ for high $c$ is that it uses the same parameter for proximity to $(a_2,b_2)$ and for the agents' reactivity to differences in expected utility. This allows us to draw a stark difference of this refinement and empirical equilibrium. Proper equilibrium is a decision-theoretical, thought experiment in which a utility maximizing agent considers the possibility that another utility maximizing agent makes a mistake. Confronted with this thought, a utility maximizing agent will determine a Nash equilibrium as implausible because it is impossible that agents who are infinitely reactive to expected utility make self-sustaining mistakes arbitrarily close to the equilibrium. By contrast, empirical equilibrium is an exercise performed by an observer based on weak payoff monotonicity, a testable property of behavior. The observer knows that  if this property is satisfied by empirical frequencies, only empirical equilibria can be approximated by data. $\qed$
 \end{example}

\section{Results}\label{Sec:Foundations}

\subsection{Approachability by payoff monotone behavior}\label{Sec:payoffmononoteapproach}

Empirical equilibrium can be equivalently defined by proximity of interior payoff monotone behavior.

\begin{theorem}\label{Thm:EE=APP-interior}\rm $\sigma\in\ENash(\Gamma(u))$ if and only if $\sigma\in \Nash(\Gamma(u))$ and there is a convergent sequence of interior payoff monotone distributions for $u$ whose limit is $\sigma$.
\end{theorem}

The characterization of the set of empirical equilibria by means of approximation of interior payoff monotone distributions simplifies its computation. In applications, the computation of this set usually requires two steps. First, one needs to identify conditions satisfied by each possible empirical equilibrium. Then, one needs to show that each equilibrium satisfying these properties is an empirical equilibrium. It is in the first step of this process that Theorem~\ref{Thm:EE=APP-interior} proves to be essential. One can identify all the implications of proximity to weakly payoff monotone behavior by assuming only proximity to interior payoff monotone behavior. The gain can be considerable in applications, for this avoids the analysis of corner cases \citep[e.g.][]{Velez-Brown-2018-EI}.

It is worth noting that computing the set of empirical equilibria of a game usually also involves a non-trivial second step in which one proves that the  candidates one identified as possible empirical equilibria are actually so. At this point it is more convenient to construct a sequence of weakly payoff monotone behavior that converges to the candidate, that is actually neither interior, nor payoff monotone \citep[e.g.][]{Velez-Brown-2018-EI}. In this sense, if one were to equivalently define empirical equilibrium based on the empirical content of payoff monotonicity, Theorem~\ref{Thm:EE=APP-interior} would still be a valuable tool in its characterization in applications.

Theorem~\ref{Thm:EE=APP-interior} indicates a form of stability of empirical equilibria. Think for instance of an equilibrium in a game that is itself a non-interior weakly payoff monotone distribution, e.g., a Nash equilibrium in which each agent plays her unique best response.\footnote{These equilibria are usually referred to as \textit{strict} \citep[c.f.,][]{Harsanyi-1973-IJGT}.} One will always conclude that the equilibrium is an empirical equilibrium by taking the respective constant sequence. Theorem~\ref{Thm:EE=APP-interior} implies that this is not the only sequence that will sustain the argument. One will always be able to find a sequence of interior payoff monotone distribution that converges to the equilibrium.

\subsection{Approachability by regular QRE}\label{Sec:rQRE}

Empirical equilibrium can be equivalently defined by approximation by the empirical content of \textit{regular Quantal Response Equilibria} \citep{Mackelvey-Palfrey-1996-JER,Goeree-Holt-Palfrey-2005-EE}.  This theory assumes that agents are noisy best responders, whose actions are responses to the vectors of expected utility. Formally, the model is parameterized by a \textit{quantal response function} (QRF),  i.e.,  for agent $i$   a function $p_i:\R^{A_i}\rightarrow \Delta(A_i)$. For each $a_i\in A_i$ and each $x\in\R^{A_i}$, $p_{ia_i}(x)$ denotes the value assigned to $a_i$ by $p_i(x)$. A QRF $p_i$ is \textit{regular} if it satisfies the following four properties \citep{Goeree-Holt-Palfrey-2005-EE}:

\medskip
\noindent- \textit{Interiority}: $p_i>0$.

\noindent- \textit{Continuity}: $p_i$ is a continuous function.

\noindent- \textit{Responsiveness}: for $x\in \R^{A_i}$, $\eta>0$, and $a_i\in A_i$, $p_{ia_i}(x+\eta1_{a_i})>p_{ia_i}(x)$.\footnote{$1_{a_i}$ denotes the vector in $\R^{A_i}$ that has $1$ in component $a_i$ and $0$ otherwise.}

\noindent- \textit{Monotonicity}: for $x\in \R^{A_i}$ and $\{a_i,\hat a_i\}\subseteq A_i$ such that $x_{a_i}>x_{\hat a_i}$, $p_{ia_i}(x)>p_{i\hat a_i}(x)$.

\medskip
A \textit{quantal response equilibrium} (QRE) of $\Gamma(u)$ with respect to a profile of QRFs, $p:=(p_i)_{i\in N}$, is a fixed point of the composition of $p$ and the expected payoff operator in $\Gamma(u)$ \citep{Goeree-Holt-Palfrey-2005-EE}, i.e., a profile of distributions $(\sigma_i)_{i\in N}$ such that for each $i\in N$, $\sigma_i=p_i((\exputil_{(\sigma_{-i},a_i)}u_i)_{a_i\in A_i})$. When $p$ is regular, we refere to the corresponding QRE also as regular.

Because regular QRFs are interior, monotone, and continuous, each regular QRE of $\Gamma(u)$ is an interior payoff monotone distribution for $\Gamma(u)$. The converse also holds. That is, regular QRE can be seen as a basis for interior payoff monotone behavior.\footnote{In a paper that circulated simultaneously with the first version of our paper, \citet{Goeree-et-al-2018} study the empirical restrictions imposed by a ranking based form of payoff monotonicity that is satisfied by each Nash equilibrium. Thus, this analysis does not produce a refinement of Nash equilibrium, which is our fundamental contribution. At a technical level, our results overlap with theirs only in that they prove a result equivalent to Lemma~\ref{Lm:Mon=rQRE}. They do not address the approximation of Nash equilibria by means of increasingly sophisticated regular QRE, nor by means of a finite dimensional form of regular QRE.}

\begin{lemma}\label{Lm:Mon=rQRE}\rm Let $\sigma$ be an interior payoff monotone distribution for $\Gamma(u)$. Then there is a regular QRF, $p$, for which $\sigma$ is a QRE for $\Gamma(u)$ with respect to $p$.
\end{lemma}

Theorem~\ref{Thm:EE=APP-interior} and Lemma~\ref{Lm:Mon=rQRE} together imply that empirical equilibrium can be equivalently defined by proximity to the empirical content of the regular QRE theory. This has two significant consequences. First, we are able to connect the empirical equilibrium refinement with the empirical content of a model that is a staple of the analysis of experimental data. Second, it opens the possibility that we provide a foundation of empirical equilibrium in terms of approximation by behavior associated with agents who are best responders in the limit. To make this precise we first need to identify the conditions under which this is so for a sequence of QRFs.

\begin{definition}\rm
A sequence of regular QRFs, $\{p^\lambda\}_{\lambda\in \N}$ is \textit{utility maximizing in the limit} if for each $u\in \Ucal$ and each convergent sequence of QREs of $\Gamma(u)$ corresponding to a  subsequence of $\{p^\lambda\}_{\lambda\in \N}$, its limit is a Nash equilibrium of $\Gamma(u)$.
\end{definition}

We can then define a refinement of Nash equilibrium in the same spirit as empirical equilibrium, but taking as basis for plausibility of behavior regular QRE for increasingly sophisticated regular QRFs.

\begin{definition}\label{Def:approachrQRE}\rm $\sigma\in \Sigma(A)$ is \textit{approachable by regular QRE that are utility maximizing in the limit in $\Gamma(u)$} if there is a sequence of regular QRF profiles, $\{p^\lambda\}_{\lambda\in\N}$, which is utility maximizing in the limit, and a corresponding convergent sequence of QREs for $\Gamma(u)$, whose limit is $\sigma$. We denote this set by $\RNash(\Gamma(u))$.
\end{definition}

Clearly, for each $u\in\Ucal$, $\RNash(\Gamma(u))\subseteq \Nash(\Gamma(u))$. The interpretation of this refinement is similar to that of empirical equilibrium. It differs in that it explicitly models approximation of Nash equilibria by behavior of agents who are infinitely sophisticated in the limit. One can argue that empirical equilibrium is a more cautious refinement, for it is based only on observables, however. Indeed, QRFs are not observables, and some of their properties, as continuity, are not refutable with finite data.

Fortunately, we do not have to choose between these notions of plausibility of behavior in games.

\begin{theorem}\label{Thm:EE=R}\rm For each $u\in\Ucal$, $\ENash(\Gamma)=\RNash(\Gamma)$.
\end{theorem}

Theorem~\ref{Thm:EE=R} allows us to alternatively describe empirical equilibrium in a way that speaks closely to the practice in experimental economics. If one expects that as players gain experience their behavior will be fit by increasingly sophisticated regular QRE, the only Nash equilibria that can be approximated by data are the empirical equilibria.

\subsection{Approachability by additive randomly disturbed payoffs models}\label{Sec:Harsanyapproach}

Two predecessors of empirical equilibrium are the Nash equilibria that are limits of behavior in \citet{Harsanyi-1973-IJGT}'s randomly disturbed payoff models for exchangeable perturbations \citep{VanDamme-1991-Springer} and logistic QRE approachable equilibria \citep{Mackelvey-Palfrey-1996-JER}. These are subrefinements of empirical equilibrium. In what follows we show that these belong to a general family of refinements of Nash equilibrium that are strict subrefinements of empirical equilibrium. As a by product, which we discuss in Sect.~\ref{Sec:Discussion}, we advance our understanding of the empirical content of monotone randomly disturbed payoff models, a topic that has received attention due to the popularity of these models for the analysis of data from economics experiments \citep{Goeree-Holt-Palfrey-2005-EE,Haile-et-al-2008,GOLMAN-2011-JET}.

We follow \citet{GOVINDAN-Reny-Robson-2003-GEB}'s construction of additive randomly disturbed payoff models. It subsumes \citet{Harsanyi-1973-IJGT}'s and all additive randomly disturbed payoff models that guarantee agents' behavior is uniquely determined by utility maximization for almost all realization of payoffs. These include the popular structural QRE of \citet{mckelvey:95geb}.

Given~$\Gamma(u)=(N,A,u)$ and a vector of independent Borel probability measures on~$\R^A$, $(\mu_i)_{i\in N}$, let $(\Gamma(u),\mu)$ be the incomplete information game  in which $\mu:= \mu_1\times\cdots\times\mu_n$ is a common prior on payoff types. Given type $\eta_i\in \R^A$ for agent $i$, her expected utility index is $a\in A\mapsto u_i(a)+\eta_i(a)$. Whenever convenient, given an agent, say $i$, whose action set is $A_i:=\{a_1,...,a_K\}$, we write a vector $x\in\R^A$ as $(x_{a_l})_{a_l\in A_{i}}$ where $x_{a_l}:=(x_{(a_{-i},a_l)})_{a_{-i}}\in A_{-i}$. The interpretation of these perturbations is that either the observer who models the strategic situation by means of game $\Gamma(u)$ does not observe the real payoffs perceived by the agents \citep{Harsanyi-1973-IJGT}, or that the agent fails to perfectly recognize the difference of payoffs between the actions and correctly maximize \citep{mckelvey:95geb}. We require throughout that:

\begin{definition}[\citealp{GOVINDAN-Reny-Robson-2003-GEB}]\rm For each $i\in N$, $\mu_i$ is \textit{purifying}, i.e., for each pair of different actions $\{a_k,a_l\}\subseteq A_i$, and each $\sigma_{-i}\in \Sigma(A)_{-i}$, $\mu_i$ assigns probability zero to the event $\eta_i(a_k,\cdot)-\eta_i(a_l,\cdot)\in\R^{A_{-i}}$ lies on any single prespecified hyperplane in $\R^{A_{-i}}$ with normal $\sigma_{-i}$. Let $\Bcal_i$ be the space purifying Borel probability measures on $\R^A$ and $\Bcal=\Bcal_1\times\dots\times\Bcal_n$.
\end{definition}

The main two models that fit into our framework are \citet{Harsanyi-1973-IJGT}'s, in which each $\mu_i$ is absolutely continuous with respect to the Lebesgue measure on $\R^A$; and \citet{mckelvey:95geb}'s structural QRE, in which perturbations are perfectly correlated across action profiles with the same action for an agent, so they can be identified with measures on $\R^{A_i}$, that are further assumed to be absolutely continuous with respect to the Lebesgue measure on $\R^{A_i}$.

Consider $i\in N$. The assumption of purifying perturbations implies that for each $\sigma_{-i}\in\Sigma(A)_{-i}$ and for $\mu_i$ almost every realization of the perturbation $\eta_i$, agent $i$ with type $\eta_i$ has a unique best response to $\sigma_{-i}$. Thus, given $\mu_i$ and $\sigma_{-i}$, the probability with which agent $i$ is observed playing a given action is uniquely defined by utility maximization. Let \[B^{u_i,\mu_i}_i(\sigma_{-i}):= (B^{u_i,\mu_i}_{ia_i}(\sigma_{-i}))_{a_i\in A_i}\in \Delta(A_i),\] be this distribution. It is easy to see that this function is continuous \citep{GOVINDAN-Reny-Robson-2003-GEB}, and that the fixed points of $\sigma\in\Sigma(A)\mapsto(B^{u_i,\mu_i}_i(\sigma_{-i}))_{i\in N}\in \Sigma(A)$, which exist by Brouwer's fixed point theorem, are the set of observable strategies in Bayesian Nash equilibria of $(\Gamma,\mu)$. We denote the set of profiles $\sigma\in\Sigma(A)$ that are induced by some  Bayesian Nash equilibrium of $(\Gamma(u),\mu)$ by $\BNE(\Gamma(u),\mu))$.

In \citet{mckelvey:95geb}'s model the best response operator can be further characterized by the function $x\in\R^{A_i}\mapsto Q_i^{\mu_i}(x)$, where
\[\sigma_{-i}\in \Sigma(A)_{-i}\mapsto B_i^{u_i,\mu_i}(\sigma_{-i})=Q_i^{\mu_i}(\exputil_{(\sigma_{-i},\cdot)}u_i)\in \Delta(A_{i}).\]
The function $Q^{\mu_i}_i$ which only depends on $N$, $A$, and $\mu_i$, and does not depend on $u$, is referred to as a \textit{structural Quantal Response Function} (sQRF) \citep{mckelvey:95geb,Goeree-Holt-Palfrey-2005-EE}. These functions are regular QRFs as defined in Sec.~\ref{Sec:rQRE} \citep{Goeree-Holt-Palfrey-2005-EE}. The sQRF most commonly used in empirical analysis of experimental data is the logistic form, $l^\lambda$, which is associated with the so-called double-exponential i.i.d perturbation \citep{Goeree-et-al-2018}, and assigns to each $a_i\in A_i$ and each $x\in\R^{A_i}$ the value,
\begin{equation}l^\lambda_{ia_i}(x):=\frac{e^{\lambda x_{a_i}}}{\sum_{\hat a_i\in A_i}e^{\lambda x_{\hat a_i}}}.\label{Equation-Logistic-QRE}\end{equation}

Additive randomly disturbed payoff models also allow us to articulate the idea that agents behavior approximates that of utility maximizers.  More precisely, consider a sequence of purifying perturbations $\{\mu^\lambda\}_{\lambda\in \N}$. We say that this sequence \textit{vanishes} when for each $i\in N$ and each neighborhood of zero, $G$, as $\lambda\rightarrow\infty$, $\mu_i^\lambda(G)\rightarrow1$. If there is a convergent sequence of Bayesian Nash equilibria of the respective games $(\Gamma(u),\mu^\lambda)$ for a sequence of vanishing perturbations, the limit of the corresponding induced observable strategies is necessarily a Nash equilibrium of $\Gamma(u)$ \citep[c.f.,][]{VanDamme-1991-Springer}.  Thus, any sequence of sQRFs whose corresponding perturbations are vanishing is utility maximizing in the limit.

If additive randomly disturbed payoff models are unrestricted, each possible observable distribution of behavior in a normal-form game is in the empirical content of this theory. More precisely, for each $\sigma\in\Sigma(A)$ and each $u\in\Ucal$, there is a purifying perturbation $\mu$ such that $\sigma\in \BNE(\Gamma(u),\mu)$ \citep{Haile-et-al-2008}. If this theory is further disciplined by convergence to Nash behavior, it produces no strict refinement of the Nash equilibrium~set.

\begin{proposition}\label{Prop:unrefutability}\rm For each $\sigma\in\Nash(\Gamma)$, there is a sequence of vanishing perturbations $\{\mu^\lambda\}_{\lambda\in \N}$ and a sequence of corresponding $\sigma^\lambda\in\BNE(\Gamma,\mu^\lambda)$ that converges to $\sigma$.
\end{proposition}

Thus, requiring proximity to the empirical content of unrestricted randomly disturbed payoff models does not refine the set of Nash equilibria.

\citet{Mackelvey-Palfrey-1996-JER} and \citet{Goeree-Holt-Palfrey-2005-EE} argue that additive randomly disturbed payoff models lack of refutability can be resolved by imposing consistency with payoff monotonicity, a phenomenon for which there is empirical support. Relatedly, two attempts have been made with the purpose of refining the set of Nash equilibria based on monotone randomly disturbed payoff models.  First, \citet{VanDamme-1991-Springer} imposes permutation invariance of perturbations. Second, \citet{Mackelvey-Palfrey-1996-JER} propose approximation by logistic QRE, i.e., restrict to a particular parametric family of perturbations. Both constructions are implicitly imposing that best responses are ordinally equivalent to expected utility.

The following theorem allows us to identify a sharp difference between these and other possible approaches based on monotone additive randomly disturbed payoff models and empirical equilibrium. It states that for each strategy space in which at least an agent has at least three actions available, one can always construct a payoff matrix so the resulting normal-form game possesses an empirical equilibrium that cannot be approximated by any additive randomly disturbed payoff model whose associated best response correspondences are weakly monotonic.

\begin{definition}\rm Let $\Mcal\subseteq \Bcal$ be the set of perturbations $\mu$ for which for each $u\in\Ucal$, each $i\in N$, and each pair $\{a_i,b_i\}$, if $B^{u_i,\mu_i}_{a_ii}(\sigma_{-i})>B^{u_i,\mu_i}_{b_ii}(\sigma_{-i})$, then $E_{(\sigma_{-i},a_i)}u_i>E_{(\sigma_{-i},b_i)}u_i$.
\end{definition}

\begin{theorem}\label{Thm:Paradox2}\rm Suppose that at least an agent has at least three actions available. Then, there is $u\in \Ucal$ for which there is $\sigma^*\in\ENash(\Gamma(u))$ and  $\varepsilon>0$ such that \[\{\sigma:||\sigma-\sigma^*||<\varepsilon\}\cap \{\sigma:\exists\mu\in\Mcal,\textrm{ s.t. }\sigma\in \BNE(\Gamma(u),\mu)\}=\emptyset.\]
\end{theorem}

Theorem~\ref{Thm:Paradox2} reveals that even though additive randomly disturbed payoff models are not refutable if unrestricted, their empirical content, if restricted by monotonicity of best responses, is restricted beyond weak payoff monotonicity. The weakly monotone behavior that is missed by the empirical content of these models can be that in the neighborhood of a Nash equilibrium. In particular, this implies that the refinements proposed by \citet{VanDamme-1991-Springer} and \citet{Mackelvey-Palfrey-1996-JER}  depend on their structural form of approximation and are strict subrefinements of empirical equilibrium for some games.

\subsection{Approachability by behavior in \citet{VanDamme-1991-Springer}'s control costs games}\label{Sec:CCosts}

We learn from Theorem~\ref{Thm:EE=R} that if behavior is weakly payoff monotone and approaches a Nash equilibrium, there is a regular QRE model that fits this behavior. The family of regular QRFs is infinitely dimensional. Thus, this theorem does not point exactly to a parametric family of models that is well specified for the analysis of experimental data that one expects to be payoff monotone.

By Theorem~\ref{Thm:Paradox2}, the most obvious candidate, the monotone structural QRE model, including all its parametric incarnations, e.g., the logistic form, is not flexible enough  to account for some payoff monotone behavior in finite games. This brings the need to find a parametric incarnation of the regular QRE model that spans all possible payoff monotonic behavior, and thus is suitable for empirical analysis. In this section we develop such a parametric model. As a byproduct we provide a constructive proof of our results in Sec.~\ref{Sec:rQRE}.

An alternative approach to rationalize deviations from utility maximizing behavior in normal-form games is \cite{VanDamme-1991-Springer}'s control costs model. Here agents are parameterized by a function that determines how difficult for the agent is to make no mistake when playing a given strategy.

A \textit{control cost function} for player $i$ is $f_i:[0,1]\rightarrow\R\cup\{\infty\}$ with the following properties.

\medskip
\noindent- $f_i$ is strictly decreasing with $f_i(0)=\infty$ and $f_i(1)=0$.

\noindent- $f_i$ is continuously differentiable on $(0,1]$.\footnote{\citet{VanDamme-1991-Springer} assumes that each $f_i$ is twice differentiable. One can easily see that Lemmas 4.2.1-4.2.5 and Theorem 4.2.6. in \citet{VanDamme-1991-Springer} go through with our weaker assumption. One needs continuous differentiability in order to apply Lagrange's theorem in Lemma 4.2.3. All other results follow from convexity and continuity. The greater generality of our model allows us to easily construct control cost functions hitting some specific targets of its derivative without matching the second derivative.}

\noindent- $f_i$ is a strictly convex function.

\medskip

Given a normal-form game $\Gamma(u)=(N,A,u)$ and a profile of control cost functions $f:=(f_i)_{i\in N}$ the associated game with control costs is that in which players are $N$, agent $i$'s action space is $\Delta(A_i)$, and payoff of action profile $\sigma\in\Sigma(A)$ is for each $i\in N$, $\exputil_{\sigma}u_i-\sum_{a_i\in A_i}f_i(\sigma_i(a_i))$. For each $\sigma_{-i}\in\Sigma(A)_{-i}$, there is a unique best response for agent $i$ in each control cost game, which solely depends on the profile of expected utility of the different actions in $\Gamma$ \citep[Lemma 4.2.1][]{VanDamme-1991-Springer}. Let $p^f_i$ be this function. One can easily see that this function is a regular QRF \citep{VanDamme-1991-Springer,Goeree-et-al-2018}.

The control cost model is less general than the regular QRE model to the extent that control cost functions impose some restrictions of behavior across different extended games, i.e., if one varies $u$. We now show that for a fixed $u$, they are behaviorally equivalent.

Let $\sigma$ be a Nash equilibrium of the control costs game associated with $\Gamma(u)$ and $f$. By interiority of each $p^f_i$ and continuous differentiability and strict convexity of each $f_i$, $\sigma$ can be characterized by the first order conditions \citep[Lemma 4.2.3][]{VanDamme-1991-Springer}: for each $i\in N$ and $\{a_l,a_k\}\subseteq A_i$,
\[\exputil_{(\sigma_{-i},a_l)}u_i-\exputil_{(\sigma_{-i},a_k)}u_i=f'_i(\sigma_i(a_l))-f'_i(\sigma_i(a_k)).\]
Thus, we can span the whole spectrum of interior payoff monotone behavior with a family of control costs functions whose derivative can be chosen for each $i\in N$ at $|A_i|-1$ given points in $(0,1)$. Indeed, an asymptote at zero stitched to a second order spline is flexible enough to interpolate any increasing slope and at the same time guarantee strict convexity and continuous differentiability.

Formally, let $\sigma$ be interior and payoff monotone for $\Gamma(u)$. Suppose for simplicity that $A_i:=\{a_1,...,a_K\}$ and $\sigma_i(a_1)\leq...\leq\sigma_i(a_K)$. Let $m_{l+1}<0$ and $m_l:= m_{l+1}-(\exputil_{(\sigma_{-i},a_{l+1})}u_i-\exputil_{(\sigma_{-i},a_{l})}u_i)$. Since $\sigma$ is payoff monotone for $\Gamma$, $m_l\leq m_{l+1}$.  The following function is strictly convex, strictly decreasing, and interpolates slopes $m_l$ and $m_{l+1}$ at the respective extremes of the interval in which it is defined: for each $y\in [\sigma_i(a_l),\sigma_i(a_{l+1})]$,
\begin{equation}f(y):= f(\sigma_i(a_{l+1}))+m_{{l+1}}(y-\sigma_{i}(a_{l+1}))+\frac{m_{l+1}-m_l}{2(\sigma_i(a_{l+1})-\sigma_i(a_l))}(y-\sigma_i(a_{l+1}))^2.\label{Eq:sticth1}\end{equation}
Let $\varepsilon>0$. Define $f_i$ on $[\sigma_i(a_K),1]$ as $y\mapsto (\varepsilon/[2(1-\sigma_i(a_K))])(y-1)^2$.\footnote{This is simply a
strictly decreasing, strictly convex, function such that $f_i(1)=0$, which parameterizes our construction.} Then, stitch the second degree polynomials on the subsequent intervals $[\sigma_i(a_l),\sigma_i(a_{l+1})]$ for $l=K-1,K-2,...,1$, i.e., define $f_i$ as in (\ref{Eq:sticth1}) in the respective intervals. Finally, stitch a strictly decreasing and strictly convex asymptote defined on $(0,\sigma_i(a_1)]$. With a view towards identifying control costs functions with further properties, one can add a calibration point $y^*_i\in(0,1)$ such that $y^*_i\not\in\sigma_i(A_i):=\{\sigma_i(a_1),...,\sigma_i(a_K)\}$ and guarantee that $f_i(y^*)<f_i(\sigma_i(a_l))+\varepsilon$ where $\sigma_i(a_l)$ is the minimum in $\{\sigma_i(a_1),...,\sigma_i(a_K)\}$ that is greater than $y^*$. Let us refer to $f:=(f_i)_{i\in N}$ so constructed as a \textit{spline calibrated by $\sigma$, $y^*:=(y^*_i)$, and $\varepsilon$}.

\begin{lemma}\label{Lm:cost=mon}\rm Let $\sigma$ be an interior payoff monotone distribution for $\Gamma(u)$; $y^*\in (0,1)^N$ such that for each $i\in N$, $y^*_i\not \in \sigma_i(A_i)$; and $\varepsilon>0$. Then $\sigma$ is a Nash equilibrium of the control costs game associated with $\Gamma(u)$ and each spline calibrated by $\sigma$, $y^*$, and $\varepsilon$.
\end{lemma}

The control costs model also allows us to identify sequences of regular QRFs that are utility maximizing in the limit. We say that a sequence of profiles of control costs functions,  $\{f^\lambda\}_{\lambda\in\N}$, \textit{vanishes}, if for each $i\in N$ and each $x\in(0,1]$,  $\lim_{\lambda\rightarrow\infty}f^\lambda_i(x)=0$.  It is well known that the behavior in games with vanishing control costs can converge only to Nash equilibria of the underlying game \citep[Theorem 4.3.1]{VanDamme-1991-Springer}.\footnote{Technically, Theorem 4.3.1 in \citet{VanDamme-1991-Springer} applies only to vanishing sequences of control cost functions of the form $\varepsilon f$ with $\varepsilon\rightarrow 0$. One can easily see that his argument extends for a general sequence of vanishing control cost functions as we define it because our assumption also implies that for each $x\in(0,1]$, $(f^\lambda_i)'(x)\rightarrow0$. We have included an explicit proof of this result in an online Appendix.\label{Footnote:vanDammeT4.3.1}} Thus, for a sequence of vanishing control cost functions its associated sequence of regular QRFs are utility maximizing in the limit.

Empirical equilibria can be alternatively characterized as the limits of behavior in games with vanishing control costs.

\begin{theorem}\label{Th:E=C}\rm  Let $\sigma\in \Nash(\Gamma(u))$ and $\{\sigma^\lambda\}_{\lambda\in\N}$ a convergent sequence of interior payoff monotone distributions for $\Gamma(u)$ whose limit is $\sigma$. Then, there is an increasing $\kappa:\N\rightarrow\N$ and $\{y^{\kappa(\lambda)}\}_{\lambda\in\N}$, where for each $\lambda\in\N$, $y^{\kappa(\lambda)}\in(0,1)^N$, such that:

\begin{enumerate}
\item  As $\lambda\rightarrow\infty$ the sequence of splines calibrated by $\sigma^{\kappa(\lambda)}$, $y^{\kappa(\lambda)}$, and $1/\kappa(\lambda)$ vanishes.
\item  For each $\lambda\in\N$, $\sigma^{\kappa(\lambda)}$ is a Nash equilibrium of the control costs game associated with $\Gamma(u)$ and the spline calibrated by $\sigma^{\kappa(\lambda)}$, $y^{\kappa(\lambda)}$, and $1/\kappa(\lambda)$.
\end{enumerate}
\end{theorem}

\section{Discussion and concluding remarks}\label{Sec:Discussion}

We have advanced our understanding of the empirical equilibrium refinement. Empirical equilibria are defined as the Nash equilibria that are the limit weakly payoff monotone behavior. Alternatively,  they can be characterized as those that are the limits of interior payoff monotone behavior (Sec.~\ref{Sec:payoffmononoteapproach}). This characterization facilitates the computation of this set in applications. Empirical equilibria are also the Nash equilibria that are the limits of regular QRE behavior for sequences of noisy best responses that approximate  best response operators in the limit (Sec.~\ref{Sec:rQRE}). A particular finite dimensional family of regular QRE, second order spline control cost games, actually span the whole empirical content of regular QRE for a fixed game (Sec.~\ref{Sec:CCosts}). This characterization provides a clear connection between this equilibrium refinement and one of the most popular structural theories for the analysis of data from economics experiments.

In contrast to regular QRE approximation, monotone best response additive randomly disturbed payoff models do not span the whole empirical content of regular QRE for a fixed game and thus are not a basis to define empirical equilibrium (Sec.\ref{Sec:Harsanyapproach}). From the point of view of the foundations of this refinement this result highlights that it is based only on a refutable theory, weak payoff monotonicity, and not on implicit or less understood cardinal or structural assumptions. By revealing this differences we also advance our understanding of randomly disturbed payoff models, in particular, the structural QRE model. As Theorem~\ref{Thm:Paradox2} is stated, it implies that for each action space in which at least an agent has at least three actions one can construct a payoff matrix that admits an empirical equilibrium that is not in the closure of the empirical content of randomly disturbed payoff models whose best response correspondences are weakly payoff monotone.  This result can be generalized for the structural QRE model. Trivially, the later requirement can be stated in terms of the corresponding structural QRFs requiring monotonicity directly on them. More interestingly,  one can actually simply require that structural QRFs admit only weakly monotone fixed points.\footnote{More precisely, one can show, by means of a fixed point argument that if a QRF violates weak payoff monotonicity, then there is a payoff matrix for which there is a QRE that violates weak payoff monotonicity.} Thus,  for each action space in which at least an agent has at least three actions one can construct a payoff matrix that admits an empirical equilibrium that is not in the closure of the empirical content of structural QRE that generate weakly payoff monotone behavior.

Thus, our results show a clear tradeoff between specification and refutability of structural QRE models. If there is no reason to rule out behavior beyond payoff monotonicity, a data analysis based on a parametric version of the structural QRE model that guarantees monotonicity may be misspecified. As a consequence, the interpretation of structural QRE estimates requires that specification of the model be empirically addressed.

\section{Appendix}

\begin{proof}[Proof of Theorem~\ref{Thm:EE=APP-interior}] Payoff monotone distributions are weakly payoff monotone. Thus we only need to prove that an empirical equilibrium is always the limit of interior payoff monotone distributions. Let $\mu$ be weakly payoff monotone for $u\in\Ucal$. Let $\varepsilon>0$. We prove that there is an interior $\gamma$ that is payoff monotone for $u$ such that $||\mu-\gamma||<\varepsilon$. This implies that $\sigma\in N(\Gamma(u))$ is the limit of a sequence of weakly payoff monotone distributions for $u$ if and only if it is the limit of a sequence of interior payoff monotone distributions for $u$.

For each $i\in N$, each $\zeta\in(0,1)$, and each profile of distributions $\beta\in \Sigma(A)$, let
\[f^\zeta_i(\beta):= (1-\zeta)\mu_i+\zeta l^\lambda((\exputil_{(\beta_{-i},a_i)}u_i)_{a_i\in A_i}),\]
where $l^\lambda$ is the logistic QRF defined in~(\ref{Equation-Logistic-QRE}).

Let $\gamma^\zeta$ be a fixed point of $f^\zeta$, that exists because $f^\zeta$ is continuous. Let $\{a_i,\hat a_i\}\subseteq A_i$. Suppose first that $\mu_i(a_i)=\mu_i(\hat a_i)$. We know that
\[l^\lambda_{a_i}((\exputil_{(\gamma^\zeta_{-i},b_i)}u_i)_{b_i\in A_i})\geq l^\lambda_{\hat a_i}((\exputil_{(\gamma^\zeta_{-i},b_i)}u_i)_{b_i\in A_i}),\]
if and only if $\exputil_{(\gamma^\zeta_{-i},a_i)}u_i\geq \exputil_{(\gamma^\zeta_{-i},\hat a_i)}u_i$. Thus, $\exputil_{(\gamma^\zeta_{-i},a_i)}u_i\geq \exputil_{(\gamma^\zeta_{-i},\hat a_i)}u_i$ if and only if $\gamma^\zeta(a_i)\geq \gamma^\zeta(\hat a_i)$. Suppose then that $\mu_i(a_i)>\mu_i(\hat a_i)$. Since $\mu$ is weakly payoff monotone for $u$, $\exputil_{(\mu_{-i},a_i)}u_i> \exputil_{(\mu_{-i},\hat a_i)}u_i$. Since as $\zeta\rightarrow0$, $\gamma^\zeta\rightarrow\mu$, there is $c>0$ such that for each $\zeta<c$, $\gamma^\zeta_i(a_i)>\gamma^\zeta_i(\hat a_i)$ and $\exputil_{(\gamma^\zeta_{-i},a_i)}u_i> \exputil_{(\gamma^\zeta_{-i},\hat a_i)}u_i$. Thus, for each pair $\{a_i,\hat a_i\}\subseteq A_i$, there is $c>0$ such that for each $\zeta<c$, $\exputil_{(\gamma^\zeta_{-i},a_i)}u_i\geq \exputil_{(\gamma^\zeta_{-i},\hat a_i)}u_i$ if and only if $\gamma^\zeta(a_i)\geq \gamma^\zeta(\hat a_i)$. Since $\Gamma(u)$ has finite action spaces, there is $c>0$ such that for each $\zeta<c$, each $i\in N$, and each pair $\{a_i,\hat a_i\}\subseteq A_i$, $\exputil_{(\gamma^\zeta_{-i},a_i)}u_i\geq \exputil_{(\gamma^\zeta_{-i},\hat a_i)}u_i$ if and only if $\gamma^\zeta(a_i)\geq \gamma^\zeta(\hat a_i)$.
\end{proof}

\begin{proof}[\textit{Proof of Proposition~\ref{Prop:unrefutability}}]Let $u\in\Ucal$ and $\sigma\in\Nash(\Gamma(u))$. We prove that there is a sequence of vanishing purifying perturbations $\{\gamma^\lambda\}_{\lambda\in\N}$ and a corresponding sequence of Bayesian Nash equilibria in the disturbed games whose observable strategy distributions converge to $\sigma$. We construct perturbations with full support satisfying \citet{Harsanyi-1973-IJGT}'s requirements that can be easily modified to induce structural QRE models.

For each Lebesgue measurable and bounded set with non-empty interior $E\subseteq \R^A$, let $U(E)$ be the uniform distribution on $E$, i.e., the normalized Lebesgue measure on it. For each $\lambda\in\N$, let $V_{\varepsilon}\subseteq\R^A$ be the open ball centered at zero with radius $\varepsilon>0$,
\[T_\lambda^{a_i}:=\{\eta_i\in \R^A:\forall a_i'\in A_i\setminus\{a_i\},\forall a_{-i}\in A_{-i},\eta_i(a_{-i},a_i)>\eta_i(a_{-i},a_i')+1/(2|A_{-i}|\lambda)\},\]
and
\[\mu^\lambda_i:=\sum_{a_i\in A_i}\sigma_i(a_i)U(V_{1/\lambda}\cap T_\lambda^{a_i}).\]
For an arbitrary full support $\gamma\in\Bcal$ and for each $\delta\in(0,1)$, let $\gamma^\lambda:=\delta\mu^\lambda+(1-\delta)\gamma$. Clearly, the sequence $\{\gamma^\lambda\}_{\lambda\in \N}$ is vanishing.

Let $\sigma^\delta$ be a fixed point of the operator $\hat \sigma\mapsto (B_i^{u_i,\gamma_i^\lambda}(\delta\sigma_{-i}+(1-\delta)\hat\sigma))_{i\in N}$ and $\hat\sigma^\delta:=\delta\sigma+(1-\delta)\sigma^\delta$. Since $\sigma\in\Nash(\Gamma(u))$, for each $\eta_i\in T_\lambda^{a_i}$,
\[\begin{array}{c}\sum_{a_{-i}\in A_{-i}}(u_i(a_{-i},a_i)+\eta_i(a_{-i},a_i))\hat\sigma_{-i}^\delta(a_{-i})-\sum_{a_{-i}\in A_{-i}}(u_i(a_{-i},a_i')+\eta_i(a_{-i},a_i'))\hat\sigma_{-i}^\delta(a_{-i})\geq
\\\sum_{a_{-i}\in A_{-i}}(1-\delta)(u_i(a_{-i},a_i)-u_i(a_{-i},a_i'))\sigma^\delta_{-i}(a_{-i})+(\eta_i(a_{-i},a_i)-\eta_i(a_{-i},a_i'))\hat\sigma^\delta(a_{-i})\geq
\\\{\sum_{a_{-i}\in A_{-i}}(1-\delta)(u_i(a_{-i},a_i)-u_i(a_{-i},a_i'))\sigma^\delta_{-i}(a_{-i})\}+1/(2|A_{-i}|\lambda)\geq
\\(1-\delta)M+\min_{i\in N}1/(2|A_{-i}|\lambda),\end{array}\]
for some $M<0$ that neither depends on $a_i$ nor on $\eta_i$. Let $\delta(\lambda)\in[1-1/\lambda,1)$ be close enough to one so the expression above is greater than zero. Then, for each realization $\eta_i\in T_\lambda^{a_i}$, agent $i$'s unique expected utility maximizer when the other agents play according to $\hat\sigma^{\delta(\lambda)}$ is $a_i$. Since $\sigma^{\delta(\lambda)}$ is a fixed point of $\hat \sigma\mapsto (B_i^{u_i,\gamma_i^\lambda}(\delta\sigma_{-i}+(1-\delta)\hat\sigma))_{i\in N}$, $\hat \sigma^{\delta(\lambda)}=B_i^{u_i,\gamma^\lambda}(\hat\sigma^{\delta(\lambda)})$. Thus, $\hat\sigma^{\delta(\lambda)}\in\BNE(\Gamma(u),\gamma^\lambda)$ and  as $\lambda\rightarrow\infty$, $\hat\sigma^{\delta(\lambda)}\rightarrow\sigma$.
\end{proof}

\begin{table}[t]
  \centering
  \begin{tabular}{cp{2.5cm}p{2.5cm}p{2.5cm}p{2.5cm}p{2.5cm}}
      &\multicolumn{5}{c}{Player 1}
  \\
  &  \multicolumn{1}{c}{$a_1$}&\multicolumn{1}{c}{$a_2$}&\multicolumn{1}{c}{$\dots$}&\multicolumn{1}{c}{$a_{K-1}$}&
  \multicolumn{1}{c}{$a_{K}$}
\\\cline{2-6}
$a_{-1}^*$&\multicolumn{1}{|c}{$5$}&\multicolumn{1}{|c|}{$5$}&\multicolumn{1}{|c|}{$\dots$}
&\multicolumn{1}{|c|}{$ 5$}&\multicolumn{1}{|c|}{$5$}
\\\cline{2-6}
$A_{-1}\setminus\{a_{-1}^*\}$&\multicolumn{1}{|c}{$1$}&\multicolumn{1}{|c|}{$2$}&\multicolumn{1}{|c|}{$\dots$}&\multicolumn{1}{|c|}{$ 2$}&\multicolumn{1}{|c|}{$4$}\\\cline{2-6}
  \end{tabular}
    \caption{Game $\Gamma(u):=(N,A,u)$, $N:=\{1,...,n\}$, $A_1:=\{a_1,...,a_K\}$ with $K\geq 3$, and $|A_{-1}|\geq 2$. The table shows the payoff of agent~$1$. Each agent $j>1$ gets a payoff of $1$ if she plays $a^*_j$ and zero otherwise.}\label{Tab:Upsilon1}
\end{table}

We now prove Theorem~\ref{Thm:Paradox2}. The game that allows us to prove this result is defined in Table~\ref{Tab:Upsilon1}. In this game each agent $i\neq 1$ has a strictly dominant action. The profile of these strictly dominant actions for these agents is~$a_{-1}^*$. Agent~$1$ is indifferent among all actions if all other agents play their strictly dominant action. Agent~$1$ has three different types of actions. Action~$a_1$, which is weakly dominated by actions~$a_2,...,a_{K-1}$, which are all payoff equivalent. All actions $\{a_1,...,a_{K-1}\}$ are weakly dominated by~$a_K$ for this agent. The essential feature of this game is that given any $\sigma_{-1}$ for which $\sigma_{-i}(a^*_{-i})<1$, the difference in agent~$1$'s expected payoff between actions~$a_K$ and~$a_{K-1}$ is greater than the difference in expected payoff between actions~$a_2$ and~$a_1$.

\begin{lemma}\label{Lm:EUpsilon1}\rm Let $\Gamma(u)$ be the game in Table~\ref{Tab:Upsilon1}. There is  $\sigma\in \Nash(\Gamma(u))$ that belongs to the closure of $\{\gamma:\gamma\textrm{ is weakly payoff monotone for }u\}$ and in which each agent $j\neq 1$ plays the strictly dominant action and $\sigma_1(a_1)<1/K<\sigma_1(a_2)= \dots=\sigma_1(a_{K-1})<\sigma_1(a_{K})$.
\end{lemma}

\begin{proof}[\textit{Proof of Lemma~\ref{Lm:EUpsilon1}}]Clearly $\Nash(\Gamma(u))$ is the set of distributions in which each $j\neq1$ plays the dominant action and agent $1$ arbitrarily randomizes. Let $\sigma$ be such that each agent $j\neq 1$ plays the strictly dominant action with certainty, and $\sigma_1(a_1)<\sigma_1(a_2)= \dots=\sigma_1(a_{K-1})<\sigma_1(a_{K})$. Then, $\sigma\in \Nash(\Gamma(u))$. Let $\lambda\in\N$ and $\sigma^\lambda$ be the convex combination that places $(1-1/\lambda)$ weight on $\sigma$ and $1/\lambda$ on a uniform distribution. Clearly as $\lambda\rightarrow\infty$, $\sigma^\lambda\rightarrow\sigma$. Thus, there is $\Lambda\in\N$ such that for each $\lambda\geq \Lambda$, $\sigma^\lambda$ is ordinally equivalent to $\sigma$. Since for each $\lambda\in \N$, $\sigma^\lambda$ is interior, $E_{(\sigma^\lambda_{-i},a_1)}u_i<E_{(\sigma^\lambda_{-i},a_2)}u_i= \dots=E_{(\sigma^\lambda_{-i},a_{K-1})}u_i<E_{(\sigma^\lambda_{-i},a_K)}u_i$, and for each $j\neq i$, if $a^*_j\in A_j$ is this agent's dominant action, $E_{(\sigma^\lambda_{-j},a_j^*)}u_j>E_{(\sigma^\lambda_{-j},a_j)}u_j$, and for each pair of actions $\{a_j,a_j'\}\subseteq A_j$ that are not dominant, $E_{(\sigma^\lambda_{-j},a_j)}u_j=E_{(\sigma^\lambda_{-j},a_j')}u_j$. Thus, $\sigma^\lambda$ is weakly payoff monotone for $u$. Thus, $\sigma$ belongs to the closure of $\{\gamma:\gamma\textrm{ is weakly payoff monotone for }u\}$. Thus, for each $1/K<\alpha<1/(K-1)$, there is $\sigma\in \Nash(\Gamma(u))$ that belongs to the closure of $\{\gamma:\gamma\textrm{ is weakly payoff monotone for }u\}$ and such that $0=\sigma_1(a_1)<\alpha=\sigma_1(a_2)= \dots=\sigma_1(a_{K-1})<(1-(K-2)\alpha)=\sigma_1(a_{K})$.
\end{proof}

Lemma~\ref{Lm:EUpsilon1} states that there is weakly payoff monotone behavior for $u$ that is arbitrarily close to an empirical equilibrium of $\Gamma(u)$ in which agent~$1$ plays actions $\{a_2,...,a_{K-1}\}$ with probability greater than $1/K$. The following proposition identifies restrictions on distributions generated by additive randomly disturbed payoff models whose best response operators are weakly payoff monotone. The proof of Theorem~\ref{Thm:Paradox2} is completed by showing, based on these restrictions, that it is impossible for these models to generate behavior close to the equilibrium identified in Lemma~\ref{Lm:EUpsilon1}.

\begin{proposition}\label{Pro:k3-or-permut--nonmonot1}\rm Let $\Gamma(u)$ be the game in Table~\ref{Tab:Upsilon1} and $\mu\in\Mcal$. Let $\sigma_{-i}\in\Delta(A_{-i})$ be such that $E_{(\sigma_{-i},a_1)}u_i<E_{(\sigma_{-i},a_2)}u_i=\dots=E_{(\sigma_{-i},a_{K-1})}u_i< E_{(\sigma_{-i},a_K)}u_i$,
and $E_{(\sigma_{-i},a_K)}u_i-E_{(\sigma_{-i},a_{K-1})}u_i>E_{(\sigma_{-i},a_{2})}u_i-E_{(\sigma_{-i},a_{1})}u_i$. Then, $B^{u_i,\mu_i}_{ia_{K-1}}(\sigma_{-i})\leq1/K$.
\end{proposition}
\begin{proof}[\textit{Proof of Proposition~\ref{Pro:k3-or-permut--nonmonot1}}] Recall that in our notation $A_i=\{1,...,K\}$. For each $k\in\{1,...,K\}$ let $v_k:= E_{(\sigma_{-i},a_k)}u_i$, fore each $x\in \R^A$, $x_k:=\sum_{a_{-i}\in A_{-i}}x_{(a_{-i},a_k)}\sigma_{-i}(a_{-i})$, and
\[X_k:=\{x\in\R^{A}:\{k\}={\arg\max}_{l=1,...,K}x_{l}\}.\]
Let $\mu\in\Mcal$. Consider $\bar u_i\in\R^{A}$ for which agent $i$ has equal payoff from each action profile. Since $\mu\in\Mcal$, $B^{\bar u_i,\mu_i}_i(\sigma_{-i})=(1/K,...,1/K)$. Thus, for each $k=1,...,K$, $\mu_i(X_k)=1/K$.

Since $\mu\in\Bcal$, for each measurable set $G\subseteq\R^{A}$,
\begin{equation}\mu_i(G)=\sum_{l=1}^K\mu_i(G\cap X_k).\label{Eq:basis1}\end{equation}
Let $G:=\{x\in\R^{A}:\{K-1\}=\arg\max_{l=1,...,K}v_l+x_{l}\}$. Then, $B^{u_i,\mu_i}_{ia_{K-1}}(\sigma_{-i})=\mu_i(G)$. Since $v_1<v_2=\dots=v_{K-1}< v_K$,
\begin{equation}\mu_i(G)=\mu_i(G\cap X_1)+\mu_i(G\cap X_{K-1}).\label{Eq:g-expression1}\end{equation}
Let $D_1:= v_{2}-v_{1}$ and $D_2:= v_K-v_{K-1}$.   Consider $u'_i\in\R^{A}$ for which $y':= E_{(\sigma_{-i},\cdot)}u_i'$ is such that $y_1'=\dots=y_{K-1}'<y_{K}':= y_{K-1}'+D_2$ (one can simply make payoffs be independent of the action of the other agents). Since $\mu_i\in\Mcal$, $B^{u_i',\mu_i}_{ia_{1}}(\sigma_{-i})=B^{u_i',\mu_i}_{ia_{K-1}}(\sigma_{-i})$. Moreover,
\[\begin{array}{l}B^{u_i',\mu_i}_{ia_{K-1}}(\sigma_{-i})=\mu_i\left(\{x\in X_{K-1}:x_{{K-1}}>x_{K}+D_2\}\right),\\
B^{u_i',\mu_i}_{ia_1}(\sigma_{-i})=\mu_i\left(\{x\in X_{1}:x_{{1}}> x_{K}+D_2\}\right).\end{array}\]
Thus,
\[\mu_i\left(\{x\in X_{K-1}:x_{{K-1}}>x_{K}+D_2\}\right)=\mu_i\left(\{x\in X_{1}:x_{{1}}> x_{K}+D_2\}\right).\]
Since $\mu_i\in\Bcal_i$ and $\mu_i(X_1)=\mu_i(X_{K-1})$,
\begin{equation}\mu_i\left(\{x\in X_{K-1}:x_{{K-1}}<x_{K}+D_2\}\right)=\mu_i\left(\{x\in X_{1}:x_{{1}}< x_{K}+D_2\}\right).\label{Eq:reverse1}\end{equation}
We claim that
\begin{equation}\begin{array}{l}\mu_i\left(\{x\in X_{1}:\max\{x_1-D_1,x_2,...,x_{K-2},x_K\}<x_{{K-1}}\}\right)\\=\mu_i\left(\{x\in X_{1}:\max\{x_1-D_1,x_2,...,x_{K-1}\}<x_{{K}}\}\right).\end{array}\label{Eq:equaldist1}\end{equation}

Consider  $u_i''\in\R^{A}$ for which $y'':= E_{(\sigma_{-i},\cdot)}u_i''$ is such that $y_1''<y_2''=\dots=y_{K}'':= y_{1}''+D_1$.  Observe that
\[B^{u_i'',\mu_i}_{ia_{K-1}}(\sigma_{-i})=\mu_i\left(\left\{x\in X_{1}:\max\{x_1-D_1,x_2,...,x_{K-2},x_K\}<x_{{K-1}}\right\}\right)+\mu_i(X_{K-1}),\]
and
\[B^{u_i'',\mu_i}_{ia_K}(\sigma_{-i})=\mu_i\left(\left\{x\in X_{1}:\max\{x_1-D_1,x_2,...,x_{K-1}\}<x_{{K}}\right\}\right)+\mu_i(X_{K}).\]
Since $\mu\in\Mcal$, $B^{u_i'',\mu_i}_{ia_K}(\sigma_{-i})=B^{u_i'',\mu_i}_{ia_{K-1}}(\sigma_{-i})$. Thus, since $\mu_i(X_{K-1})=\mu_i(X_{K})$, (\ref{Eq:equaldist1}) follows.

Since $D_2\geq D_1$, by monotonicity of measures with respect to set inclusion
\[\begin{array}{l}\mu_i\left(\{x\in X_{1}:x_{1}<x_K+D_2\}\right)\\
\geq \mu_i\left(\{x\in X_{1}:x_{1}<x_{K}+D_1\}\right)\\
\geq \mu_i\left(\{x\in X_{1}:\max\{x_1-D_1,x_2,...,x_{K-1}\}<x_{{K}}\}\right).\end{array}\]
Replacing (\ref{Eq:reverse1}) and (\ref{Eq:equaldist1}) in the first and last expressions of the inequality above yields,
\[\begin{array}{l}\mu_i\left(\{x\in X_{K-1}:x_{K-1}<x_K+D_2\}\right)\\
\geq \mu_i\left(\{x\in X_{1}:\max\{x_1-D_1,x_2,...,x_{K-2},x_K\}<x_{{K-1}}\}\right).\end{array}\]
Now,
\[\mu_i(G\cap X_{K-1})=\mu_i(X_{K-1})-\mu_i\left(\{x\in X_{K-1}:x_{K-1}<x_K+D_2\}\right),\]
and by monotonicity of measures with respect to set inclusion,
\[\begin{array}{rl}\mu_i(G\cap X_{1})&=\mu_i\left(\{x\in X_{1}:\max\{x_1-D_1,x_2,...,x_{K-2},x_K+D_2\}<x_{K-1}\}\right)
\\
&\leq\mu_i\left(\{x\in X_{1}:\max\{x_1-D_1,x_2,...,x_{K-2},x_K\}<x_{K-1}\}\right).\end{array}\]
Thus,
\[\mu_i(G\cap X_{1})+\mu_i(G\cap X_{K-1})\leq \mu_i(X_{K-1})=1/K.\]
Thus, by (\ref{Eq:g-expression1}),
\[B^{u_i,\mu_i}_{ia_{K-1}}(\sigma_{-i})=\mu_i(G)=\mu_i(G\cap X_{1})+\mu_i(G\cap X_{K-1})\leq1/K.\]
\end{proof}

\begin{proof}[Proof of Theorem~\ref{Thm:Paradox2}]Let $\Gamma(u)$ be the game in Table~\ref{Tab:Upsilon1}. By Lemma~\ref{Lm:EUpsilon1}, there is $\sigma^*\in \Nash(\Gamma(u))$ that belongs to the closure of $\{\gamma:\gamma\textrm{ is weakly payoff monotone for }u\}$ in which each agent $j\neq 1$ plays the strictly dominant action and $\sigma_1^*(a_1)<1/K<\sigma_1^*(a_2)=\dots=\sigma_1^*(a_{K-1})<\sigma_1^*(a_{K})$. Thus, there is $\varepsilon>0$ for which for each $\sigma\in\Delta$ such that $||\sigma-\sigma^*||<\varepsilon$, $\sigma_1(a_1)<1/K<\sigma_1(a_2)$. Let $\sigma$ that is weakly payoff monotone for $u$ and $||\sigma-\sigma^*||<\varepsilon$. Since $\sigma_1(a_1)<1/K<\sigma_1(a_2)$ and $\sigma$ is weakly payoff monotone for $u$, we have that $E_{(\sigma_{-i},a_1)}u_i<E_{(\sigma_{-i},a_2)}u_i$. Thus, $\sigma_{-i}(a^*_{-i},a_1)<1$, for otherwise $E_{(\sigma_{-i},a_1)}u_i=E_{(\sigma_{-i},a_2)}u_i$. Thus, $E_{(\sigma_{-i},a_2)}u_i-E_{(\sigma_{-i},a_1)}u_i=(1-\sigma_{-i}(a^*_{-i},a_1))>0$ and $E_{(\sigma_{-i},a_K)}u_i-E_{(\sigma_{-i},a_{K-1})}u_i=2(1-\sigma_{-i}(a^*_{-i},a_1))$. Thus,  $E_{(\sigma_{-i},a_K)}u_i-E_{(\sigma_{-i},a_{K-1})}u_i>E_{(\sigma_{-i},a_2)}u_i-E_{(\sigma_{-i},a_1)}u_i$. Thus, there is no $\mu\in\Mcal$ such that $\sigma\in \BNE(\Gamma(u),\mu)$, for otherwise by Proposition~\ref{Pro:k3-or-permut--nonmonot1}, $\sigma_1(a_2)\leq 1/K$. Thus,
$\{\sigma:||\sigma-\sigma^*||<\varepsilon\}\cap \{\sigma:\exists\mu\in\Mcal,\textrm{ s.t. }\sigma\in \BNE(\Gamma(u),\mu)\}=\emptyset$.
\end{proof}

\begin{proof}[\textit{Proof of Lemma~\ref{Lm:cost=mon}}]Let $f$ be a control cost function. Since for each $f_i(0)=\infty$, any equilibrium of the game associated with $\Gamma(u)$ and $f$ is interior \citep[Lemma 4.2.1]{VanDamme-1991-Springer}. Moreover, $\sigma$ is an equilibrium of the control costs game associated with $\Gamma(u)$ and $f$ if and only if for each $i\in N$, and each pair $\{a_l,a_k\}\subseteq A_i$, \citep[Theorem 4.2.6.]{VanDamme-1991-Springer}
\begin{equation}\exputil_{(\sigma_{-i},a_l)}u_i-\exputil_{(\sigma_{-i},a_k)}u_i=f'_i(\sigma_i(a_l))-f'_i(\sigma_i(a_k)).\label{Eq:charac-control-cost-eq}\end{equation}
Thus, given $\sigma$, one can construct $f$ for which $\sigma$ is an equilibrium of the game associated with $\Gamma(u)$ and $f$ if for each $i$ one construct $f_i$ for which (\ref{Eq:charac-control-cost-eq}) holds.

Let $\sigma$ be interior and payoff monotone for $\Gamma(u)$. Let $i\in N$. Assume that $A_i:=\{a_1,...,a_K\}$. For each $k=1,...,K$, let $u_k:= \exputil_{(\sigma_{-i},a_k)}u_i$. Suppose without loss of generality that $u_1\leq...\leq u_K$. Since $\sigma$ is payoff monotone for $\Gamma(u)$, $\sigma_i$ is ordinally equivalent to $u$. In particular, $\sigma_i(a_1)\leq...\leq\sigma_i(a_K)$. Let $\{y_1,...,y_L\}:=\{y\in(0,1):\exists k\in\{1,...,K\},\sigma_i(a_k)=y\}$ be the set of values that $\sigma_i$ takes. For each $l=1,...,L$, let $u_l:= \exputil_{(\sigma_{-i},a_k)}u_i$ for $a_k$ such that $\sigma_i(a_k)=y_l$. Suppose without loss of generality that $0<y_1<...<y_L<1$. Since  $\sigma_i$ is ordinally equivalent to $u$, $u_1<...<u_L$. Fix $\varepsilon>0$ and $0<y_0<y_1$ and let $y_{L+1}=1$ ($y_0$ is not necessary for the proof of the lemma; we introduce it in order to use the argument later in the proof of Theorem~\ref{Th:E=C}). Let $m_{L+1}:=0$; $m_{L}:=-\varepsilon$; for each  $l=1,...,L-1$, let $m_l:= m_{l+1}-(u_{l+1}-u_l)$; and $m_0=m_1-\varepsilon$.

Consider $f_i$ that assigns to each $y\in(0,1]$ the value
\[f_i(y):=\left\{\begin{array}{ll}m_{L+1}(y-y_{L+1})+\frac{m_{L+1}-m_L}{2(y_{L+1}-y_L)}(y-y_{L+1})^2&\textrm{ if }y\in[y_{L},y_{L+1}].
\\
f_i(y_{l+1})+m_{l+1}(y-y_{l+1})+\frac{m_{l+1}-m_l}{2(y_{l+1}-y_l)}(y-y_{l+1})^2&\textrm{ if }y\in[y_{l},y_{l+1}),\ l=0,...,L-1.
\\
f_i(y_{0})-m_0y_0^2\frac{1}{y}&\textrm{ if }y\in(0,y_0)
\end{array}\right.\]
The function $f_i$ is a second order spline stiched to a hyperbole in the interval $(0,y_0)$. It is continuous and has continuous derivative on $(0,1]$. Indeed, it coincides with infinitely differentiable functions in $(0,1]$ with the exception of $y_0,...,y_L$. In these points, the derivative from the left and from the right match, so its derivative on $(0,1]$ is well defined and continuous. The function is strictly decreasing, for its derivative is always negative on $(0,1)$. It satisfies $\lim_{y\rightarrow 0}f_i(y)=\infty$. The second derivative of $f_i$ is well defined and positice on $(0,1]\setminus\{y_0,...,y_L\}$. Thus, the function is strictly convex.

Finally, let $\{a_s,a_k\}\subseteq A_i$. Let $l$ and $t$ be such that $\sigma_i(a_l)=y_l$ and $\sigma_i(a_k)=y_t$. Suppose without loss of generality that $l<t$. Then,
\[\begin{array}{ll}f_i'(\sigma_i(a_s))-f_i'(\sigma_i(a_k))&=m_l-m_{r}\\
&(m_l-m_{l+1})+(m_{l+1}-m_{l+2})+...+(m_{r-1}-m_r)
\\
&=(u_l-u_{l+1})+(u_{l+1}-u_{l+2})+...+(u_{r-1}-u_r)\\
&=u_l-u_r\\
&=\exputil_{(\sigma_{-i},a_s)}u_i-\exputil_{(\sigma_{-i},a_k)}u_i.\end{array}\]
\end{proof}

\begin{proof}[\textit{Proof of Theorem~\ref{Th:E=C}}]Let $\sigma\in\ENash(\Gamma(u))$ and $\{\sigma^\lambda\}_{\lambda\in\N}$ a sequence of interior payoff monotone distributions converging to $\sigma$. Let $i\in N$, $A_i:=\{a_1,...,a_K\}$. By passing to a subsequence if necessary we can suppose without loss of generality that for each $\lambda\in\N$, $\sigma_i^\lambda(a_1)\leq ...\leq \sigma_i^\lambda(a_K)$. By convergence of $\{\sigma^\lambda\}_{\lambda\in\N}$ we have that $\sigma_i(a_1)\leq ...\leq \sigma_i(a_K)$. Let $k$ be the lowest index for which $\exputil_{(\sigma_{-i},a_{k})}u_i=\exputil_{(\sigma_{-i},a_{K})}u_i$. Then, there is $\eta>0$ such that for each $l<k$, $\exputil_{(\sigma_{-i},a_{k})}u_i-\exputil_{(\sigma_{-i},a_{l})}u_i>\eta$.

Let $\{f^\lambda\}$ be a sequence of control cost functions constructed as in the proof of Lemma~\ref{Lm:cost=mon} for parameters $\varepsilon<1/\lambda$ and $y_0<1/\lambda$. Let $\{a_r,a_s\}\subseteq\{a_k,...,a_K\}$. Suppose that $\sigma_i(a_k)=0$. By passing to a subsequence if necessary we can assume that for each $\lambda\in\N$, $\sigma_i^\lambda(a_{k})<1/\lambda$ and $f_i^\lambda(\sigma^\lambda_i(a_k))<2/\lambda$ (this subsequence is constructed by first selecting a distribution for which the utility among the best responses of agent $i$ to $\sigma_{-i}$ have utility differences at most $1/\lambda$ and the probabilities are in the required ranges; then for that distribution construct the function $f_i^\lambda$ with $\varepsilon<1/\lambda$ and $y_0<1/\lambda$). Since $\lim_{\lambda\rightarrow\infty}\sigma_i^\lambda(a_k)=0$, for each $y\in(0,1]$, $\lim_{\lambda\rightarrow\infty} f^\lambda_i(y)=0$. Repeating the argument for each agent we construct a subsequence $\{f^\lambda\}$ that vanishes.

Suppose that  $\sigma_i(a_k)>0$ and $k>1$. By passing to a subsequence if necessary we can assume that for each $\lambda\in\N$, $\sigma_i^\lambda(a_{k-1})<1/\lambda$, $(f_i^\lambda)'(\sigma^\lambda_i(a_{k-1}))=(f_i^\lambda)'(\sigma^\lambda_i(a_{k}))+\exputil_{(\sigma^\lambda_{-i},a_{k-1})}u_i
-\exputil_{(\sigma^\lambda_{-i},a_{k-1})}u_i<4/\lambda$, $(f_i^\lambda)'(\sigma^\lambda_i(a_k))=1/\lambda+\exputil_{(\sigma^\lambda_{-i},a_{k})}u_i
-\exputil_{(\sigma^\lambda_{-i},a_{K})}u_i<2/\lambda$. Then for each $\lambda$ we can construct a control cost function $g^\lambda$ for which $\sigma^\lambda$ is an equilibrium of the game associated with $\Gamma(u)$ and $g^\lambda$, as in the proof of Lemma~\ref{Lm:cost=mon} for parameters $\varepsilon<1/\lambda$ and $y_0<1/\lambda$, and including a calibration point $y=1/\lambda$ strictly which is strictly in between $\sigma^\lambda_i(a_{k-1})$ and $\sigma^\lambda_i(a_k)$ and for which we can set a slope of $g^\lambda_i$ equal to $3\lambda$. Then, the sequence of control cost functions $\{g^\lambda\}_{\lambda\in\N}$ is such that for each $y\in(0,1]$, $\lim_{\lambda\rightarrow\infty} g^\lambda_i(y)=0$. The result concludes as in the previous case.

Finally, suppose that Suppose that  $\sigma_i(a_k)>0$ and $k=1$. Note that the slope of each $f^\lambda_i$ in the set $[y^\lambda_0,1]$ is bounded above by $2\varepsilon+\exputil_{(\sigma^\lambda_{-i},a_{1})}u_i-\exputil_{(\sigma^\lambda_{-i},a_{K})}u_i$. Since $\lim_{\lambda\rightarrow\infty}y^\lambda_0=0$ and $\lim_{\lambda\rightarrow\infty}\exputil_{(\sigma^\lambda_{-i},a_{1})}u_i-\exputil_{(\sigma^\lambda_{-i},a_{K})}u_i=0$, then for each $y\in(0,1]$, $\lim_{\lambda\rightarrow\infty} f^\lambda_i(y)=0$. The result concludes as in the previous case.
\end{proof}

\begin{proof}[\textit{Proof of Lemma~\ref{Lm:Mon=rQRE} and Theorem~\ref{Thm:EE=R}}]Since best response operators in games with control costs are regular QRFs, Lemma~\ref{Lm:Mon=rQRE} is a corollary of Lemma~\ref{Lm:cost=mon}. When a sequence of control cost functions vanishes, the corresponding best response operators are utility maximizing in the limit (See footnote~\ref{Footnote:vanDammeT4.3.1}). Thus, Theorem~\ref{Thm:EE=R} is a corollary of  Theorem~\ref{Th:E=C}.
\end{proof}


\bibliography{ref-BC}

\newpage
\section*{Appendix not for publication}

\begin{lemma}[\citealp{VanDamme-1991-Springer}]\rm\label{Lm:vanishingcontrol-nash}Let $\{f^\lambda\}_{\lambda\in \N}$ be a sequence of control costs functions that vanishes and $\{\sigma^\lambda\}_{\lambda\in\N}$ a corresponding convergent sequence of the control cost game associated with $\Gamma(u)$ and $f^\lambda$. Then, $\{\sigma^\lambda\}_{\lambda\in\N}$ converges to a Nash equilibrium of $\Gamma(u)$.
\end{lemma}

\begin{proof}[\textit{Proof of Lemma~\ref{Lm:vanishingcontrol-nash}}]Let $\{\sigma^\lambda\}_{\lambda\in\N}$ be a convergent sequence such that for each $\lambda\in\N$, $\sigma^\lambda$ is an equilibrium of the control cost game associated with $\Gamma(u)$ and $f^\lambda$. Let $\sigma:= \lim_{\lambda\rightarrow\infty}\sigma^\lambda$. Let $i\in N$ and $a_i\in A_i$ be a best response to $\sigma_{-i}$ for agent $i$ in $\Gamma(u)$. Suppose that $a_k\in A_i$ is not a best response to $\sigma_{-i}$ for agent $i$. We prove that $\sigma_i(a_k)=0$. Since as $\lambda\rightarrow\infty$, $\sigma^\lambda\rightarrow\sigma$, we also have that $\sigma_i^\lambda(a_i)\rightarrow \sigma_i(a_i)$, $\sigma_i^\lambda(a_k)\rightarrow \sigma_i(a_k)$, $\exputil_{(\sigma^\lambda_{-i},a_i)}u_i\rightarrow \exputil_{(\sigma_{-i},a_i)}u_i$, and $\exputil_{(\sigma^\lambda_{-i},a_k)}u_i\rightarrow \exputil_{(\sigma_{-i},a_k)}u_i$. Thus, there is $\Lambda\in\N$ such that for each $\lambda\geq \Lambda$, $\sigma_i^\lambda(a_k)\leq \sigma_i^\lambda(a_i)$. Suppose first that $\sigma_i(a_i)=0$. Since $\sigma_i^\lambda(a_i)\rightarrow 0$, $\sigma_i^\lambda(a_k)\rightarrow0$. Suppose then that $\sigma_i(a_i)>0$.  By \citep[Theorem 4.2.6]{VanDamme-1991-Springer}, for each $\lambda\in\N$,
\[\exputil_{(\sigma^\lambda_{-i},a_l)}u_i-\exputil_{(\sigma^\lambda_{-i},a_k)}u_i=(f^\lambda)_i'(\sigma^\lambda_i(a_l))-(f^\lambda)_i'(\sigma^\lambda_i(a_k)).\]
The left side of the expression above converges to a positive number. Since $\sigma^\lambda_i(a_l)\rightarrow\sigma_i(a_l)>0$ and $\{f^\lambda\}_{\lambda\in\N}$ vanishes, $(f^\lambda)_i'(\sigma^\lambda_i(a_l))\rightarrow 0$. Thus, $\sigma^\lambda_i(a_k))\rightarrow0$, for otherwise there is a subsequence of $\{\sigma^\lambda_i(a_k))\}_{\lambda\in\N}$ that converges in the interior of $(0,1]$. If this is so the right side of the equation above converges to zero. This is a contradiction.
\end{proof}

The following proposition formally states our claims in Example~\ref{Ex:Gamma_c}

\begin{proposition}\rm\label{Prop:Gamma_c}Consider the game $\Gamma_c$ in Table~\ref{Tab:Gamma_c}. Then, for each $c>0$,
  \[\begin{array}{l}\Nash(\Gamma_c)=\{(a_1,b_1), (a_2,b_2), (a_3,b_3)\},
  \\\TNash(\Gamma_c)=\UNash(\Gamma_c)=\{(a_1,b_1), (a_2,b_2)\},\\\PNash(\Gamma_c)=\{(a_1,b_1)\}.\end{array}\]
Moreover,
\[\ENash(\Gamma_c)=\left\{\begin{array}{ll}\{(a_1,b_1)\}&\textrm{if }\min\{c_1,c_2\}\leq 1,\\\{(a_1,b_1),(a_2,b_2)\}&\textrm{Otherwise. }\end{array}\right.\]
\end{proposition}

\begin{proof}We first prove that\linebreak $\Nash(\Gamma_c)=\{(a_1,b_1), (a_2,b_2), (a_3,b_3)\}$.  Let $c:=(c_1,c_2)$ such that $c_1>0$ and $c_2>0$. One can easily see that the action profiles $(a_1,b_1)$, $(a_2,b_2)$, and $(a_3,b_3)$ are the only pure strategy Nash equilibria of $\Gamma_c$. Now, let $\sigma\in\Nash(\Gamma_c)$. If $\sigma_1(a_2)>0$, then $\sigma_2(b_3)=0$. Then, $\sigma_1(a_3)=0$. It follows that either $\sigma$ is equal to $(a_1,b_1)$ or $(a_2,b_2)$. Symmetry implies the same is true when $\sigma_2(b_2)>0$. Thus, suppose that $\sigma_1$ is not a pure strategy. Suppose that $\sigma_1(a_1)>0$, $\sigma_1(a_2)=0$, and $\sigma_2(b_2)=0$. Then $\sigma_2(b_3)=0$. Thus, $\sigma=(a_1,b_1)$. A symmetric argument shows that if $\sigma_2(b_1)>0$, $\sigma_1(a_2)=0$, and $\sigma_2(b_2)=0$, then $\sigma=(a_1,b_1)$.

It is well known that at each perfect equilibrium no agent plays a weakly dominated strategy. Clearly, $a_3$ and $b_3$ are weakly dominated for players~$1$ and~$2$, respectively. Thus, $\TNash(\Gamma_c)\subseteq \{(a_1,b_1), (a_2,b_2)\}$. Now, let $t:=\min\{c_1,c_2\}$, $\varepsilon:= \min\{t,t/(3c_1),t/(3c_2),1/3\}$, and for each $\lambda\in\N$, $\sigma^\lambda$ be the strategy profile for which $\sigma^\lambda_1(a_1):=\varepsilon c_2/(2\lambda t)$ and $\sigma^\lambda_1(a_3):=\varepsilon/(\lambda t)$; and $\sigma^\lambda_2(b_1):=\varepsilon c_1/(2\lambda t)$ and $\sigma^\lambda_2(b_3):=\varepsilon/(\lambda t)$. Then,
\[\begin{array}{l}E_{(\sigma^\lambda_2,a_1)}u^c_1=\varepsilon c_1/(2\lambda t)-(7+c_1)\varepsilon/(\lambda t)=-7\varepsilon/(\lambda c_1)-\varepsilon c_1/(2\lambda t),\\
E_{(\sigma^\lambda_2,a_2)}u^c_1=-7\varepsilon/(\lambda c_1),
\\E_{(\sigma^\lambda_2,a_2)}u^c_1=-(7+c_1)\varepsilon/(2\lambda)-7(1-\varepsilon/\lambda-\varepsilon/(\lambda c_1))-7\varepsilon/(\lambda c_1).\end{array}\]
Thus, $a_2$ is the unique best response to $\sigma^\lambda_2$ for agent $1$. Symmetry implies that $b_2$ is the unique best response to $\sigma^\lambda_1$ for agent $2$. Since $\sigma^\lambda_1$ places probability at most $1/\lambda$ in both $a_1$ and $a_3$; $\sigma^\lambda_2$ places probability at most $1/\lambda$ in both $b_1$ and $b_3$; and as $\lambda\rightarrow\infty$, $\sigma^\lambda\rightarrow(a_2,b_2)$, we have that $(a_2,b_2)\in \TNash(\Gamma_c)$.

Let $\Lambda>2$ be such that for each $\lambda\geq\Lambda$, $1-1/(2\lambda)-1/(3\lambda^2)>\max\{c_1/(3\lambda^2),c_2/(3\lambda^2),1/\lambda\}$. Let $\lambda\geq\Lambda$ and $\sigma^\lambda$ be the symmetric profile of strategies such that $\sigma^\lambda_1(a_2):= 1/(2\lambda)$ and $\sigma^\lambda_1(a_3):= 1/(3\lambda^2)$. Thus, $E_{(\sigma^\lambda_2,a_1)}u^c_1-E_{(\sigma^\lambda_2,a_2)}u^c_1=1-1/(2\lambda)-1/(3\lambda^2)-c_1/(3\lambda^2)>0$. Clearly, $E_{(\sigma^\lambda_2,a_2)}u^c_1>E_{(\sigma^\lambda_2,a_3)}u^c_1$. Similarly, $E_{(\sigma^\lambda_1,b_1)}u^c_2-E_{(\sigma^\lambda_1,b_2)}u^c_2>0$ and $E_{(\sigma^\lambda_1,b_2)}u^c_2>E_{(\sigma^\lambda_1,b_3)}u^c_2$. Since $\sigma^\lambda_1(a_1)>\sigma^\lambda_1(a_2)/\lambda$, $\sigma^\lambda_1(a_2)>\sigma^\lambda_1(a_3)/\lambda$, $\sigma^\lambda_2(b_1)>\sigma^\lambda_2(b_2)/\lambda$, and $\sigma^\lambda_2(b_2)>\sigma^\lambda_2(b_3)/\lambda$; and as $\lambda\rightarrow\infty$, $\sigma^\lambda\rightarrow(a_1,b_1)$, we have that $(a_1,b_1)\in \PNash(\Gamma_c)\subseteq \TNash(\Gamma_c)$.
\end{proof}

\end{document}